\pdfoutput=1
\documentclass[11pt]{article}
\usepackage[margin=1.25in]{geometry}
\usepackage{setspace}
\usepackage{amsmath,amssymb,amsthm,enumitem,booktabs,array}
\usepackage[colorlinks=true,linkcolor=blue,citecolor=blue,urlcolor=blue]{hyperref}
\usepackage[round]{natbib}
\newtheorem{theorem}{Theorem}[section]
\newtheorem{proposition}[theorem]{Proposition}
\newtheorem{lemma}[theorem]{Lemma}
\newtheorem{corollary}[theorem]{Corollary}
\newtheorem{assumption}{Assumption}
\newtheorem{definition}[theorem]{Definition}
\theoremstyle{remark}
\newtheorem{remark}[theorem]{Remark}
\newtheorem{example}[theorem]{Example}

\newcommand{\E}{\mathbb{E}}
\newcommand{\Pp}{\mathbb{P}}
\newcommand{\Var}{\mathrm{Var}}
\newcommand{\R}{\mathbb{R}}

\newcommand{\one}{\mathbf{1}}
\newcommand{\tx}{\widetilde{x}}
\newcommand{\supp}{\mathrm{supp}}
\newcommand{\kap}{\kappa_{\mathcal C}}
\newcommand{\hbeta}{\widehat\beta}

\title{Exact Inference in Fixed-Effect Regressions\\ with Concentrated Identifying Variation\thanks{Code and data reproducing all computations, including the cycle-packing implementation and the verification suite, are available in \texttt{PanelAdequacy.jl} (\url{https://github.com/profsms/PanelAdequacy.jl}) and the companion \textsf{R} package \texttt{panelcert} (\url{https://github.com/profsms/panelcert}).}}

  \author{Stanis\l{}aw M.~S.~Halkiewicz\\ \small Lule\aa{} University of Technology, Lule\aa{}, Sweden\\ \small \texttt{stashal@o2.pl}}
  \date{October 2026}

\newcounter{ssec}

\begin{document}
\maketitle

\begin{abstract}
In fixed-effect regressions with many groups, a handful of observations can carry most of the identifying variation left after fixed effects are removed. Conventional $t$-tests can then reject a true null more than half the time. This paper specializes the exact $t$-test of Pouliot to such designs. Contrasts chosen from the design eliminate the nuisance coefficients exactly; sign flips give finite-sample valid tests under independent symmetric errors, or compatible independent symmetric error blocks, allowing arbitrary heteroskedasticity. In two-way designs, the annihilating space is the cycle space of the bipartite observation graph and contains all within variation. This makes the choice of contrasts a graph problem. The observable capture ratio $\kap\in[0,1]$ measures retained identifying variation and, under diffuse homoskedastic conditions, gives the efficiency relative to an oracle test. Structured packing trades capture against the number and balance of usable signs. The matched-weight test is also asymptotically valid without symmetry when independent centered contrast scores satisfy a fourth-moment negligibility condition. In the Grunfeld investment regression, where one observation carries $73.9\%$ of the realized squared score, 32 contrasts capture $\kap=0.627$ and give a symmetry-based exact $95\%$ confidence interval of $[0.150,\,0.450]$.
\end{abstract}


\section{Introduction}\label{sec:intro}

After fixed effects absorb most treatment variation, a handful of observations can carry what remains. Across 1{,}309 instrumental-variable regressions in 30 published papers, \citet{Young2022} finds that one cluster or observation accounts on average for $0.18$ of residualized instrument variation and up to $0.70$ in the most leveraged papers. This paper treats that concentration, rather than aggregate variation, as the regularity condition for Gaussian inference and develops a finite-sample exact alternative.

Consider the linear model with high-dimensional fixed effects,
\begin{equation}\label{eq:model}
Y_i = x_i \beta + d_i'\gamma + \varepsilon_i, \qquad i = 1, \ldots, n,
\end{equation}
with $x_i$ a scalar treatment and $d_i$ the $d_n$ fixed-effect dummies. Writing $\tx = M_D x$ for the treatment residualized on the fixed effects and $V_n = \tx'\tx$ for the identifying variation that survives, the least-squares $t$-statistic is, up to studentization, the score $\sum_i (\tx_i/\sqrt{V_n})\varepsilon_i$: a weighted sum of the errors whose weights $a_{ni}=\tx_i/\sqrt{V_n}$ satisfy $\sum_i a_{ni}^2=1$. A standard condition for a Gaussian approximation robust to non-Gaussian error distributions is that no single weight dominate that sum, the Lindeberg-type condition
\begin{equation}\label{eq:lambda}
\lambda_n \;:=\; \max_{i \le n}\, \tx_i^2 / V_n \;\longrightarrow\; 0,
\end{equation}
routinely imposed in the many-covariate and network-regression literatures \citep{CJN2018, CattaneoJanssonNewey2018ET, JochmansWeidner2019, MikushevaSolvsten2024}. In words, $\lambda_n$ is the maximum leverage of the residualized treatment: the largest share of the identifying variation $V_n$ carried by any single observation. When that share does not vanish, a Gaussian approximation need not be valid. Proposition~\ref{ex:rade} makes the fully concentrated boundary precise: the limiting null law depends on the error distribution, and no fixed critical value has rejection probability exactly $\alpha$ for every symmetric law with a common variance. Conservative fixed cutoffs do exist: $\alpha^{-1/2}$ bounds the worst-case pointwise asymptotic rejection probability by $\alpha$, but is much larger than the Gaussian cutoff. The randomization construction below supplies finite-sample validity by adapting to the realized scores.

Fixed effects are dummy designs, so persistent leverage heterogeneity is not
an isolated finite-sample accident.  \citet{AnatolyevYaskov2017} study the
diagonals of many-instrument and many-regressor projection matrices.  Equal
group sizes produce homogeneous diagonals, whereas unequal groups and both
stand-alone and interacted dummy regressors can retain non-trivial asymptotic
variation; their mixing-versus-common-factor intuition also anticipates why a
large block can remain dominant here.  Their diagonal-dispersion statistic is
not $\lambda_n$---it concerns the whole regressor projector rather than the
maximum residualized-treatment share---but it supplies structural precedent
for treating heterogeneous leverage as an expected feature of saturated FE
designs.  Their calibrations reach dispersion of $0.25$ and higher and, for a
different modified $J$-statistic, imply asymptotic-variance distortion up to
$50\%$ at unit excess kurtosis; these are quantitative precedents, not size
claims for the $t$-test studied below.

But the difficulty just described is not the end of the story. One can build linear combinations of the data from the design alone, using only the treatment and the fixed-effect structure rather than the outcome, that eliminate the fixed effects exactly. Once nothing has to be estimated, flipping the signs of those combinations gives a test that is exactly valid at any sample size. Formally, a \emph{nuisance-annihilating contrast} is a design-measurable $q\in\R^n$ with $q'D=0$. Under $\beta=\beta_0$, $q'(Y-x\beta_0)=q'\varepsilon$ contains neither fixed effects nor estimated quantities. Disjointly supported contrasts are invariant to independent sign flips when errors are independent and symmetric, with arbitrary heteroskedasticity, yielding exact finite-sample inference (Theorem~\ref{thm:exact}). In one-way designs this specializes to familiar sign-change and cluster-flip procedures \citep{CRS2017,CSS2021,Toulis2022,HemerikGoemanFinos2020,DeSantisEtAl2025}; exact annihilation removes their homogeneity or nuisance-estimation requirements, while symmetry remains essentially unavoidable for exactness \citep{DutzZhang2026}.

The matched-weight construction has a direct antecedent in
\citet{Pouliot2026}. His exact $t$-test projects the tested regressor off the
span of all group-transformed nuisance regressors. With independent sign flips
on the selected supports, that projection separates into local nuisance
projections. For the projection contrasts used here, its unstudentized statistic
is exactly $\sum_c s_cb_cU_c$; taking the absolute value gives our two-sided
statistic (Proposition~\ref{prop:exactt}). Pouliot also establishes the principle
of exact inference under invariance and asymptotic validity without it. We give
a direct support-sign-flip version in Proposition~\ref{prop:asymmetry}: the same
matched-weight test remains asymptotically valid for independent centered,
possibly asymmetric contrast scores under the proposition's fourth-moment
negligibility condition. In particular, variance-share negligibility suffices
when standardized fourth moments are uniformly bounded.
The contribution here is the fixed-effect geometry and the construction and
cost of usable supports, rather than a new general randomization principle.

The substance of the paper is the two-way case (Section~\ref{sec:cycles}), and its main point needs no new vocabulary to state. In worker--firm, firm--time, or student--teacher designs, the observations form a bipartite multigraph, and the linear combinations that annihilate the two-way fixed effects turn out to be exactly the \emph{cycles} of that graph: alternating $\pm 1$ contrasts around closed walks, including ``digons'' formed by parallel edges (repeated matches). So the mobility that identifies the treatment effect and the contrasts that make exact inference possible are the same object, and the entire within variation lives in that cycle space: formally, $V_n = \|\Pi_{\mathcal Z}\, x\|^2$, the squared norm of the projection of the treatment onto the cycle space (Proposition~\ref{prop:cyclespace}). Everything that follows is a consequence of this fact: which cycles to use is a packing problem, the capture ratio $\kap$ introduced below is the observable price of exactness, and that price runs into a granularity floor set by how many independent signs the chosen cycles carry.

The exactness of Theorem~\ref{thm:exact} does not depend on leverage: it holds whether or not $\lambda_n$ vanishes. What it costs is power, and that cost is observable from the design alone, before the outcome is ever seen. The theorem itself requires only disjoint supports, not cycle-shaped $\pm1$ contrasts. A \emph{support system} consists of disjoint edge sets $A_1,\ldots,A_C$, each carrying the normalized local projection $v_j\propto\Pi_{\mathcal Z_{A_j}}x$. Its observable \emph{capture ratio} is
\[
\kap \;=\; \frac{1}{V_n}\sum_{j \le C} \big\|\Pi_{\mathcal Z_{A_j}} x\big\|^2 \;\in\; [0, 1],
\]
which becomes $\sum_c(v_c'x)^2/V_n$ for edge-disjoint $\pm1$ cycles. In the diffuse benchmark, $\kap$ is Pitman efficiency relative to the infeasible oracle Gaussian test, so $\kap^{-1/2}$ is the asymptotic standard-error price (Theorem~\ref{thm:power}).

Capture rises when supports merge and reaches $\kap=1$ with one projection contrast per biconnected block (Proposition~\ref{prop:dom}; Theorem~\ref{thm:blockproj}). Yet the randomization group carries one sign per support: for the two-sided statistic the attainable floor is $2^{1-C}$, and merging can violate the power condition that no support dominate captured variation. Fine supports restore granularity but may lose capture. Proposition~\ref{prop:orbitmax}(c) formalizes this \emph{capture--granularity trade-off}: enlarging the flip group restricts its orbit span and weakly lowers the capture ceiling.

Cycle packing supplies interpretable, granular supports. Section~\ref{sec:packing} constructs digons, edge-disjoint firm-pair four-cycles with optimal nested pairing (Remarks~\ref{prop:auto} and~\ref{lem:nested}), and recursive cycles after graph contraction (Proposition~\ref{prop:contract}). On the public \citet{KSS2020} extract it attains $\kap=0.51$ for a match-level treatment versus $0.26$ for greedy packing and $0.91$ versus $0.88$ for a time-varying covariate. The standard-error prices are $1.05\times$--$1.40\times$; Section~\ref{sec:overlap} locates the remaining gap to one inside a dominant block.

Monte Carlo evidence (Section~\ref{sec:numerics}) completes the picture. On a two-way design with concentrated identifying variation ($\lambda_n = 0.29$, effective sample size $6.0$), the degrees-of-freedom-corrected $t$-test rejects a true null at rates up to $58.5\%$, and the HC2 $t$-test up to $33.4\%$, under symmetric heteroskedastic errors; the cycle test's empirical size is $4.9$--$5.0\%$ in every configuration, consistent with the finite-sample size control of Theorem~\ref{thm:exact}. Oracle size correction removes the conventional tests' spurious power advantage: across the two heteroskedastic designs, cycle-test power is $49.5$--$52.0\%$, versus $27.5$--$49.6\%$ for df-$t$ and $23.5$--$51.1\%$ for HC2. In a diffuse design the cycle test tracks the $\kap$-efficiency benchmark, with the remaining finite-$C$ gap reported explicitly. Section~\ref{sec:grunfeld} runs the concentrated Grunfeld specification end to end, while Section~\ref{sec:kss_application} demonstrates the sparse construction at worker--firm scale. Appendix~H retains a diffuse dense-panel check.

\paragraph{What is new.} Sign-flipping is classical \citep{LehmannRomano2005,HemerikGoeman2018}, and nuisance removal by projection against a group's transformed design is the general exact-$t$ construction of \citet{Pouliot2026}. The matched-weight projection test here is its block-sign specialization, not a separate general testing method. The contributions are the concentration analysis and its fixed-critical-value qualification (Proposition~\ref{ex:rade}); the two-way cycle-space characterization and its implications for constructing exact tests (Proposition~\ref{prop:cyclespace}); observable diffuse-regime Pitman efficiency $\kap$ (Theorem~\ref{thm:power}); and a structure-exploiting packing algorithm built on the contraction principle of Proposition~\ref{prop:contract}. Section~\ref{sec:overlap} characterizes the capture--granularity trade-off. Proposition~\ref{prop:asymmetry} records the asymptotic robustness of the support-sign procedure with explicit sufficient conditions.

\section{Background on the concentrated regime}\label{sec:concentrated}

Proposition~\ref{thm:conv} recalls the classical convolution limit for triangular arrays, and Corollary~\ref{cor:stud} describes the effect of consistent studentization. Proposition~\ref{ex:rade} shows that, at the fully concentrated boundary, no fixed cutoff attains the nominal rejection probability across all symmetric error laws with a common variance. The complementary diffuse regime $\lambda_n\to0$ is studied  by \citet{FEMeasurementError}.

Throughout, the analysis is conditional on the design: $x \in \R^n$ and the fixed-effect dummy matrix $D \in \R^{n \times d_n}$ are treated as fixed, $M_D = I_n - D(D'D)^- D'$, $\tx = M_D x$, $V_n = \tx'\tx > 0$, and $\lambda_n$ is defined in \eqref{eq:lambda}. Under \eqref{eq:model} and the null-centered outcome $r = Y - x\beta_0$, the score is $\tx' r = \tx'\varepsilon + V_n(\beta - \beta_0)$; we study its null behavior through the weights $a_{ni} := \tx_i / \sqrt{V_n}$, which satisfy $\sum_i a_{ni}^2 = 1$ and $\max_i a_{ni}^2 = \lambda_n$.

\begin{assumption}\label{ass:err}
Conditional on $(x, D)$, the errors $\varepsilon_1, \ldots, \varepsilon_n$ are independent and identically distributed with law $F$, $\E[\varepsilon_i] = 0$, and $\Var(\varepsilon_i) = \sigma^2 \in (0, \infty)$.
\end{assumption}

The i.i.d.\ assumption in this section is for transparency of the limit only; the exact tests of Section~\ref{sec:contrasts} dispense with it entirely.

\begin{proposition}[Convolution limit; classical]\label{thm:conv}
Let Assumption~\ref{ass:err} hold, and suppose there exist an integer $J \ge 0$ and constants $a_1, \ldots, a_J$ such that $a_{nj} \to a_j$ for each $j \le J$ (after relabeling observations) while $\max_{i > J} |a_{ni}| \to 0$. Write $c := \sum_{j \le J} a_j^2 \in [0, 1]$. Then
\[
S_n := \sum_{i=1}^n a_{ni}\, \varepsilon_i \;\Longrightarrow\; \sum_{j=1}^{J} a_j\, \varepsilon_j^{*} \;+\; \sqrt{(1-c)}\;\sigma Z,
\]
where $\varepsilon_1^*, \ldots, \varepsilon_J^*$ are i.i.d.\ draws from $F$ and $Z \sim N(0,1)$ is independent of them.
\end{proposition}

\begin{proof}
See Appendix~\ref{app:proofs}.
\end{proof}

When a few weights remain non-negligible, the corresponding errors persist in the limit, while the remaining observations contribute an independent Gaussian component. The resulting distribution need not be Gaussian.

\begin{corollary}[Studentization does not restore normality]\label{cor:stud}
In the setting of Proposition~\ref{thm:conv}, let $\widehat\sigma_n$ be any estimator with $\widehat\sigma_n \to_p \sigma$ (for instance, the degrees-of-freedom-corrected residual standard deviation when $d_n/n \to \rho < 1$ and $\E \varepsilon^4 < \infty$). Then $S_n / \widehat\sigma_n \Rightarrow \sigma^{-1}\big[\sum_{j \le J} a_j \varepsilon_j^* + \sqrt{1 - c}\, \sigma Z\big]$, which is standard normal for all $F$ if and only if $c = 0$.
\end{corollary}

\begin{proof}
Immediate from Proposition~\ref{thm:conv} and Slutsky. If $c = 0$ the limit is $Z$. If $c > 0$, at least one $a_j$ is nonzero. Were the displayed convolution Gaussian, Cram\'er's characterization of the normal law, applied repeatedly to its independent summands, would force every nondegenerate $a_j\varepsilon_j^*$ to be Gaussian and hence force $F$ itself to be Gaussian. Thus a non-Gaussian $F$ gives a non-Gaussian limit.
\end{proof}

The next result distinguishes exact calibration from conservative validity. In the fully concentrated regime the limiting null law varies with the error distribution, so one fixed cutoff cannot attain the nominal rejection probability throughout the class. A larger cutoff can nevertheless control its worst-case pointwise asymptotic rejection probability.

\begin{proposition}[No universally exact fixed critical value under concentration]\label{ex:rade}
Let $\mathcal F_{\mathrm{sym}}(\sigma^2)$ be the class of symmetric distributions on $\R$ with mean zero and variance $\sigma^2>0$. Consider the fully concentrated case $J=1$, $a_1=1$ of Proposition~\ref{thm:conv} under Corollary~\ref{cor:stud}, so $T_n:=S_n/\widehat\sigma_n\Rightarrow\varepsilon^*/\sigma$, with $\varepsilon^*\sim F$. For every fixed $\varkappa>0$, the family $F_p$ in the proof satisfies
\[
 R(F_p,\varkappa):=\lim_n\Pp_{F_p}(|T_n|>\varkappa)=p,
 \qquad 0<p<\min\{1,\varkappa^{-2}\}.
\]
Hence no fixed cutoff has limiting rejection probability $\alpha\in(0,1)$ for every $F$ in the class. Nonetheless,
\[
 \sup_{F\in\mathcal F_{\mathrm{sym}}(\sigma^2)}
 \limsup_n\Pp_F(|T_n|>\varkappa)
 =\min\{1,\varkappa^{-2}\}.
\]
Thus any cutoff $\varkappa\ge\alpha^{-1/2}$ satisfies $\limsup_n\Pp_F(|T_n|>\varkappa)\le\alpha$ for every $F$ in the class.
\end{proposition}

\begin{proof}
See Appendix~\ref{app:proofs}.
\end{proof}

At $\alpha=0.05$, the cutoff $\alpha^{-1/2}\approx4.47$ exceeds the Gaussian $z_{0.975}\approx1.96$. It can be conservative. For example, in the limiting experiment with a standardized signal of $3.5$ and Rademacher noise, the statistic is $3.5\pm1$ with equal probabilities; the $\sqrt{20}$ cutoff has power $1/2$. Quantitative power comparisons therefore require a specified alternative and error law.

\begin{remark}[Rademacher boundary]\label{rem:rade}
At $p=1$ the family becomes the Rademacher law on $\{\pm\sigma\}$. Weak convergence gives limiting rejection probability $1$ for $\varkappa<1$ and $0$ for $\varkappa>1$. It does not determine the limit at $\varkappa=1$, an atom of the absolute limiting statistic. For example, take an intercept, $x=(1,0,\ldots,0)'$, independent Rademacher errors, and the consistent scale $\widehat\sigma_n=1$. Then
\[
 a_{n1}=\sqrt{(n-1)/n},\qquad
 a_{ni}=-1/\sqrt{n(n-1)}\quad(i>1).
\]
Although $S_n\Rightarrow\varepsilon_1$, for every even $n$ symmetry of the odd-length Rademacher sum gives $\Pp(|S_n|>1)=1/2$. Thus the strict tail of the limiting law cannot be substituted at this boundary. The open family used in Proposition~\ref{ex:rade} has no atom at the chosen cutoff and avoids this issue.
\end{remark}

The randomization tests of Section~\ref{sec:contrasts} adapt to the realized error distribution automatically and exactly, supplying the kind of non-tabulated validity that Proposition~\ref{ex:rade} motivates.

\begin{remark}[Diagnostics]\label{rem:diag}
Both $\lambda_n$ and the Herfindahl index $H_n := \sum_i (\tx_i^2/V_n)^2$, with
effective sample size $N_{\mathrm{eff}} := 1/H_n$, are observable from the design
and should be reported alongside saturated FE estimates; \citet{Young2022}
computes $\lambda_n$-type shares across published work and finds them large.
Under conditional heteroskedasticity, write $\sigma_i^2 :=
\Var(\varepsilon_i\mid x,D)$ and define the population score-variance shares
and their concentration by
\[
 w_{ni}^{\Omega}:=\frac{\tx_i^2\sigma_i^2}{\sum_j\tx_j^2\sigma_j^2},
 \qquad
 \lambda_n^{\Omega}:=\max_i w_{ni}^{\Omega},
 \qquad
 H_n^{\Omega}:=\sum_i (w_{ni}^{\Omega})^2.
\]
For independent errors with uniformly integrable standardized squares,
$\lambda_n^{\Omega}\to0$ is the score-level negligibility condition behind the
heteroskedastic CLT. The design condition $\lambda_n\to0$ implies it when the
conditional variances are uniformly bounded above and away from zero, but need
not do so under unrestricted heteroskedasticity. A directly reportable realized
analogue is
\[
 \widehat w_{ni}^{\mathrm{score}}
 :=\frac{\tx_i^2\widehat u_i^2}{\sum_j\tx_j^2\widehat u_j^2},\qquad
 \widehat\lambda_n^{\mathrm{score}}:=\max_i\widehat w_{ni}^{\mathrm{score}},
 \quad
 \widehat N_{\mathrm{eff}}^{\mathrm{score}}
 :=\left\{\sum_i(\widehat w_{ni}^{\mathrm{score}})^2\right\}^{-1}.
\]
Because it uses one realization and fitted residuals, this is a warning
diagnostic rather than, without further conditions, a consistent estimator of
$\lambda_n^{\Omega}$. Neither the design nor score diagnostic originates here,
but the partition they induce is the one that matters: when the relevant
effective sample size is large, conventional leverage-corrected
inference applies \citep{MacKinnonWhite1985, CJN2018, KSS2020, Jochmans2022};
when it is small, Gaussian approximations can fail, and Proposition~\ref{ex:rade}
shows that at the fully concentrated boundary no fixed cutoff attains the
nominal rejection probability for every symmetric law. Conservative fixed
cutoffs remain possible; Section~\ref{sec:contrasts} supplies an adaptive exact test. The diagnostic and the boundary between the
two remedies are computed from the same design objects. In particular,
$\lambda_n$, $H_n$, $N_{\mathrm{eff}}$, $\rho_n=d_n/n$, and the capture
$\kap$ of any specified packing are functions of $(x,D)$ alone and require no
outcome; only the realized score diagnostic uses $Y$.
\end{remark}

\begin{remark}[Levels, logs, and partial prevalence evidence]\label{rem:levelslogs}
We screened 49 prespecified positive regressor pairs from ten public panel
distributions, always comparing levels and logs on the same observations and
without selecting variables after seeing the diagnostics. Logs lower
$\lambda_n$ in 37 of 49 pairs ($75.5\%$): the paired mean falls from $0.170$
to $0.097$. Capture generally moves the other way because diffuse variation is
easier to cover with disjoint supports. Representative design-only results are:
\begin{center}
\begin{tabular}{lrrrr}
\toprule
& \multicolumn{2}{c}{$\lambda_n$} & \multicolumn{2}{c}{$\kap$}\\
Regressor & levels & logs & levels & logs\\
\midrule
Grunfeld investment & 0.348 & 0.078 & 0.571 & 0.857\\
Grunfeld capital & 0.254 & 0.188 & 0.668 & 0.750\\
Munnell GSP & 0.116 & 0.027 & 0.664 & 0.925\\
Munnell private capital & 0.081 & 0.025 & 0.720 & 0.913\\
Munnell highways & 0.032 & 0.062 & 0.860 & 0.911\\
\bottomrule
\end{tabular}
\end{center}
This is partial prevalence evidence, not a claim that logging is uniformly
regularizing: highways are a counterexample, and realized score concentration
falls under logs in only about half of the outcome-bearing pairs. The full
paired table, including null results and variables with no within variation, is
in the replication archive.
\end{remark}

\section{Exact inference via nuisance-annihilating contrasts}\label{sec:contrasts}

We now drop Assumption~\ref{ass:err} entirely. The model is \eqref{eq:model}, conditional on $(x, D)$, and the target is $H_0 : \beta = \beta_0$ with $\gamma$ an unrestricted nuisance.

\begin{definition}[Annihilating contrast system]\label{def:contrast}
A collection $q_1, \ldots, q_C \in \R^n$, each $q_c \ne 0$, is an \emph{annihilating contrast system} for the design $(x, D)$ if (i) each $q_c$ is a measurable function of $(x, D)$ only; (ii) $q_c' D = 0$ for every $c$; and (iii) the supports $S_c := \supp(q_c)$ are pairwise disjoint. Write $v_c := q_c / \|q_c\|$ and $b_c := v_c' x$.
\end{definition}

\begin{assumption}[Blockwise symmetric errors]\label{ass:sym}
There is a partition of $\{1, \ldots, n\}$ into blocks $B_1, \ldots, B_K$, measurable with respect to $(x, D)$, such that, conditional on $(x, D)$: (i) the subvectors $\varepsilon_{B_1}, \ldots, \varepsilon_{B_K}$ are mutually independent; (ii) each $\varepsilon_{B_k}$ is centrally symmetric, $\varepsilon_{B_k} \overset{d}{=} -\varepsilon_{B_k}$; and (iii) each contrast support $S_c$ is a union of blocks.
\end{assumption}

No moment or identical-distribution condition appears in Assumption~\ref{ass:sym}: variances may differ arbitrarily across blocks and need not exist. In the leading cases: for one-way designs, blocks are groups, contrasts are within-group score vectors $q_g = $ within-group demeaned $x$ restricted to group $g$, and (ii) requires central symmetry of each group's error vector, permitting arbitrary within-group dependence; for two-way designs with observation-level independence, blocks are singletons and (ii) requires each $\varepsilon_i$ to be symmetric about zero. The essential role of symmetry is not an artifact of our construction: \citet{DutzZhang2026} show that for wide classes of one-sample nulls, invariance assumptions of exactly this type are necessary for the existence of finite-sample exact randomization tests.

\begin{lemma}[Sign-flip invariance]\label{lem:inv}
Let Assumption~\ref{ass:sym} hold, let $q_1, \ldots, q_C$ be an annihilating contrast system, and define the contrast scores $U_c := v_c'(Y - x \beta_0)$. Under $H_0$, for every $s \in \{-1, +1\}^C$,
\[
(U_1, \ldots, U_C) \;\overset{d}{=}\; (s_1 U_1, \ldots, s_C U_C), \qquad \text{conditional on } (x, D).
\]
\end{lemma}

\begin{proof}
See Appendix~\ref{app:proofs}.
\end{proof}

Under the null, sign invariance supplies the reference distribution for the following randomization test.

\begin{theorem}[Finite-sample exactness]\label{thm:exact}
Let the conditions of Lemma~\ref{lem:inv} hold, let $T : \R^C \to \R$ be any measurable test statistic, and let $\mathcal G = \{-1, +1\}^C$. Define the randomization $p$-value
\[
p \;:=\; \frac{1}{|\mathcal G|} \sum_{s \in \mathcal G} \one\big\{ T(s_1 U_1, \ldots, s_C U_C) \ge T(U_1, \ldots, U_C) \big\} .
\]
Then $\Pp(p \le \alpha \mid x, D) \le \alpha$ under $H_0$, for every $n$, every $\alpha \in (0, 1)$, and every error distribution satisfying Assumption~\ref{ass:sym}. The same holds for the Monte Carlo version based on $B$ i.i.d.\ uniform draws $s^{(1)}, \ldots, s^{(B)}$ from $\mathcal G$, drawn independently of $U$ conditional on $(x,D)$, with $p = \big(1 + \#\{b : T(s^{(b)} \circ U) \ge T(U)\}\big) / (1 + B)$.
\end{theorem}

\begin{proof}
See Appendix~\ref{app:proofs}.
\end{proof}

Compare the observed pattern of contrast signs against every possible sign pattern consistent with the null (or a random sample of them), using any test statistic at all. The fraction of patterns that score at least as extreme as the one actually observed is a valid $p$-value at every sample size, with no asymptotics, no estimated nuisance parameters, and no restriction on how concentrated the design is. This standard group-invariance guarantee is the starting point for the design problem: how much identifying variation the contrasts retain, and how to retain it with enough usable signs.

\begin{remark}[Statistic and confidence intervals]\label{rem:stat}
Our default statistic is the matched-weight score $T(u) = |\sum_c b_c u_c|$, motivated by the alternative: for general $\beta$, $U_c = b_c (\beta - \beta_0) + v_c'\varepsilon$, so $b_c$ is the signal loading of contrast $c$. The self-normalized statistic $|\sum_c b_c u_c| / (\sum_c b_c^2 u_c^2)^{1/2}$ gives exactly the same sign-flip $p$-value: its denominator is constant on the orbit. Set this statistic to zero if its denominator is zero, when all matched scores vanish and $p=1$. This identity also supplies the asymptotic interpretation in Proposition~\ref{prop:asymmetry}. Because each $U_c(\beta_0) = v_c'Y - b_c \beta_0$ is affine in $\beta_0$, the test inverts to a confidence set by a one-dimensional search; with the unstudentized statistic the acceptance region is an interval. Writing $T(u) = |\sum_c b_c u_c|$, the entire randomization orbit is affine in $\beta_0$,
\[
T\big(s \circ U(\beta_0)\big) \;=\; \Big| \sum_c s_c b_c v_c'Y \;-\; \beta_0 \sum_c s_c b_c^2 \Big| ,
\]
so the two inner products are computed once per sign pattern and the $p$-value curve over an arbitrarily fine grid of $\beta_0$ costs nothing further. Appendix~H inverts the test on a grid of $24{,}001$ points with $B = 99{,}999$ flips in seconds.
\end{remark}

\subsection{Relation to Pouliot's exact $t$-test}\label{sec:exactt}

\citet{Pouliot2026} removes nuisance coefficients by projecting off every
transformed nuisance regressor allowed by the chosen group. His Lemma~A.1
covers finite groups of orthogonal linear transformations, including the block
sign flips used here. The following specialization, pointed out by Pouliot,
identifies the statistic and makes its rank calculation local.

\begin{proposition}[Block-sign specialization of the exact $t$-test]\label{prop:exactt}
Let $A_1,\ldots,A_C$ be disjoint selected supports, let
$A=\bigcup_c A_c$, and restrict the data to $A$. Write $D_c=D_{A_c,:}$,
$x_c=x_{A_c}$, and
\[
 Q_c=I_{|A_c|}-D_c(D_c'D_c)^-D_c',\qquad z_c=Q_cx_c.
\]
For $s\in\{\pm1\}^C$, let $g_s=\bigoplus_c s_cI_{|A_c|}$.
The nuisance orbit span and its orthogonal-complement projector are
\[
 \mathcal H_A=\operatorname{span}\{\operatorname{col}(g_sD_A):s\in\{\pm1\}^C\},
 \qquad Q_A=\Pi_{\mathcal H_A^\perp}=\bigoplus_c Q_c,
\]
and
\begin{equation}\label{eq:exactt-rank}
 \operatorname{rank}(Q_A)=\sum_c\bigl(|A_c|-\operatorname{rank}(D_c)\bigr).
\end{equation}
Embed $v_c=z_c/\|z_c\|$ in $\R^n$ when $z_c\ne0$, omitting zero terms,
and write $b_c=v_c'x$, $U_c=v_c'(Y-x\beta_0)$. Then Pouliot's
unstudentized statistic on the selected observations satisfies
\begin{equation}\label{eq:exactt-equality}
 x_A'Q_Ag_s(Y_A-x_A\beta_0)=\sum_c s_cb_cU_c.
\end{equation}
Thus its two-sided sign-flip test is the matched-weight projection-contrast test.
\end{proposition}

\begin{proof}
See Appendix~\ref{app:proofs}.
\end{proof}

The equality uses the local projection contrast. It also holds for an ordinary
cycle contrast whenever the local annihilating space is one-dimensional, since
then $Q_cx_c=b_cv_c$ regardless of the orientation of $v_c$. An arbitrary
non-projection contrast on a higher-dimensional support need not give the same
statistic; replacing it by the local projection is the upgrade in
Proposition~\ref{prop:dom}. Unselected observations are discarded in the
comparison, not silently given an extra flip. The rank formula uses the actual
local nuisance ranks, which may differ across supports and are usually below
the number of global FE columns. It also applies after augmenting $D$ with
continuous nuisance regressors, as in the Grunfeld application.

\subsection{Asymptotic validity without symmetry}\label{sec:asymmetry}

Exactness under symmetry and asymptotic validity without symmetry are compatible
guarantees. This is a central feature of \citet{Pouliot2026}; his Theorem~3
establishes it for the block-permutation construction under its own design and
rate conditions. It cannot be invoked verbatim for a growing FE design with
support sign flips. Here independence across supports permits a direct proof,
using the conditional CLT for a Rademacher-weighted sum. The normalizer is the
sum of squared null contrast scores, not the residual-based variance estimator
in Pouliot's block-permutation implementation.

\begin{proposition}[Asymptotic validity of support sign flips]\label{prop:asymmetry}
Consider a sequence of design-measurable annihilating contrast systems with
$C=C_n\to\infty$. Conditional on $(x,D)$, suppose
$W_{nc}=v_c'\varepsilon$ are independent and mean zero; symmetry is not
required. Put
\[
 \xi_{nc}=b_cW_{nc},\qquad
 \mathcal V_n=\sum_c\E[\xi_{nc}^2\mid x,D]>0,
 \qquad
 L_n=\frac{\sum_c\E[\xi_{nc}^4\mid x,D]}{\mathcal V_n^2}.
\]
If $L_n\to0$, the full-enumeration sign-flip test with
$T(u)=|\sum_cb_cu_c|$ has conditional null rejection probability tending to
$\alpha$ for every $\alpha\in(0,1)$. The same holds for the plus-one Monte
Carlo test when the number of independent sampled flips $B_n\to\infty$.
Under Assumption~\ref{ass:sym}, both tests retain the finite-sample validity of
Theorem~\ref{thm:exact}, without the moment or negligibility conditions above.
\end{proposition}

\begin{proof}
See Appendix~\ref{app:proofs}.
\end{proof}

A sufficient condition is uniformly bounded standardized fourth moments of the
nondegenerate $W_{nc}$ and
\[
 \frac{\max_c b_c^2\operatorname{Var}(W_{nc}\mid x,D)}{\mathcal V_n}\to0.
\]
With a common score variance this reduces to (P1) below; under
heteroskedasticity it is the variance shares, not just the treatment loadings,
that must vanish. Independent centered error blocks with supports that respect
the blocks give the required independent centered scores. Neither arbitrary
cross-support dependence nor a dominant variance contribution is covered.
In particular, self-normalization does not repair asymmetry in a persistently
concentrated score system. The finite-sample guarantee there still requires
sign invariance. No new algorithm is needed for the robust regime: the
self-normalized and unstudentized statistics rank every sign pattern identically.

\subsection{Which dependence structures are covered}\label{sec:dep}

Assumption~\ref{ass:sym}(i)--(ii) requires independence and symmetry of the error blocks. Part (iii) requires the contrast supports to respect those blocks, a condition that can be checked from the design. For worker--firm panels, the relevant blocks may be workers, firms, matches, or periods.

\paragraph{The combinatorial part.} Let $\mathcal P$ be a partition of the
observations into dependence cells---the level at which a practitioner would
cluster. Condition (iii) says each contrast support must be a union of cells.
Since the supports are ours to choose, this is a constraint on granularity, and
it binds differently for each $\mathcal P$.

In one sentence: a support is compatible with a chosen clustering level exactly when it never splits a dependence cell, so whether a given packing tolerates worker-, firm-, or period-level dependence can be read off directly from which cells its supports respect, before any data are seen.

\begin{lemma}[Compatible supports in two-way designs]\label{lem:compat}
Let the observations be the edges of the bipartite multigraph $G$ of
Section~\ref{sec:cycles}. A support $A_j$ satisfies Assumption~\ref{ass:sym}(iii)
for a partition $\mathcal P$ iff it is $\mathcal P$-measurable. Hence:
\emph{(a)} at observation level ($\mathcal P$ = singletons) every support system
is compatible; \emph{(b)} at match level ($\mathcal P$ = parallel-edge classes)
$A_j$ must contain all or none of each match's edges, so every cycle through
multiplicity-one matches qualifies, and a match of multiplicity $k$ taken whole
carries a local cycle space of dimension $k-1$; \emph{(c)} at $\mathcal A$-vertex
(worker) level $A_j$ must use each worker's edges all or none, so \emph{every
four-cycle built from two-period movers is automatically compatible}, as is every
cycle produced by the contraction of Proposition~\ref{prop:contract}, since such
cycles traverse each mover via both of her edges; \emph{(d)} at
$\mathcal B$-vertex (firm) level no single cycle is compatible unless every firm
it touches has degree two, so compatibility forces supports at the level of
groups of whole firms; \emph{(e)} at period level the statement is (d) with the
roles of the vertex classes exchanged.
\end{lemma}

\begin{proof}
See Appendix~\ref{app:proofs}.
\end{proof}

Part (c) is the useful one: the packing algorithm of Section~\ref{sec:packing} was designed for capture, but it delivers worker-clustering robustness for free, because its contrasts are built from whole movers rather than individual spells. Part (d) is the binding restriction, and Appendix~H prices it: moving from observation-level to firm-level blocks on a real design takes the contrast count from $45$ to $8$ while leaving capture essentially unchanged, and moving to country-level blocks reaches $\kap=1$ with only $C=3$ supports. The group has order $2^3$, so $1/8$ is the generic orbit bound; for the default two-sided absolute statistic, global sign reversal duplicates every statistic and the smallest full-enumeration $p$-value is actually $2^{1-3}=0.25$. Either way, no level-$5\%$ full-enumeration test exists. This is the capture--granularity trade-off appearing as a clustering question rather than a power question.

Additive shocks at fixed-effect levels are removed by the contrasts, as the following remark shows.

\begin{remark}[Random-effects dependence at a fixed-effect level is free]\label{prop:reann}
Let $B \subseteq \{1,\ldots,n\}$ be a group whose indicator satisfies
$\one_B\in\mathrm{col}(D)$, and suppose the errors contain an additive component
$a\one_B$, where $a$ is an arbitrary scalar random variable. Then $q'\varepsilon$
does not depend on $a$, for every annihilating contrast $q$: because
$q\perp\mathrm{col}(D)$ and $\one_B\in\mathrm{col}(D)$, $q'\one_B = 0$ and
$q'\varepsilon = a\,q'\one_B + q'e = q'e$, an identity rather than an
approximation since $q$ is a function of $(x,D)$ alone. The same argument
applies term by term to any random vector in $\mathrm{col}(D)$: every such
component is annihilated identically.
Consequently, if the remaining idiosyncratic errors $e$ have independent symmetric
coordinates across the whole sample, the observation-level contrast system is exact
--- no coarsening, no loss of capture, and no distributional or magnitude
restriction on the annihilated random effects. This does not require
Assumption~\ref{ass:sym} to hold for the raw errors $\varepsilon = a\one_B + e$
themselves: a shared, nondegenerate $a$ makes the singleton coordinates of
$\varepsilon$ within $B$ mutually dependent, so Assumption~\ref{ass:sym}(i) can
fail for $\varepsilon$. What licenses exactness is that every contrast score
equals $q'\varepsilon = q'e$ identically, and Lemma~\ref{lem:inv}'s proof goes
through verbatim with $\varepsilon$ replaced by $e$, whose singleton coordinates
are independent and symmetric by hypothesis --- $a$ never has to satisfy any
symmetry, independence, or moment condition, because it never appears in the
statistic at all.
\end{remark}

This covers the random-effects (equicorrelated, ``cluster'') structure that
motivates most applied clustering at a level that is itself a fixed effect. In
the two-way design of Section~\ref{sec:cycles} this includes both vertex classes;
it includes matches only if match indicators are explicitly in
$\mathrm{col}(D)$. What
Lemma~\ref{lem:compat}(b)--(e) is really about, then, is the \emph{non-equicorrelated}
residue---serial correlation of a firm's deviations from its own mean, not the
mean itself---and the figures above price exactly that residue.

Dependence whose joint law is not invariant under the selected support-level sign flips falls outside the exactness argument. General interactive-factor models $\varepsilon_{it}=\gamma_i f_t+e_{it}$ need not have this invariance. There is, however, an important sufficient condition: conditional on the design, let the loadings $\gamma_i$ be independent and symmetric about zero, the idiosyncratic firm vectors $e_i$ be independent and centrally symmetric, and these collections and the factor vector $f$ be mutually independent. Conditional on $(x,D,f)$, the firm error vectors are then independent and centrally symmetric. Theorem~\ref{thm:exact} applies to firm-compatible supports under this additional conditioning, and averaging over $f$ gives validity conditional on $(x,D)$. No factor estimation is required. The interactive-factor simulation in Appendix~F satisfies these conditions. This argument does not validate observation-level supports that split firm blocks, nor arbitrary factor models with asymmetric loadings or additional dependence.

\paragraph{The symmetry part.} Central symmetry of each block is the price of
exactness and is close to unavoidable \citep{DutzZhang2026}, but it is weaker
than it looks once the contrasts are antithetic, and in the leading designs it
can be replaced by an exchangeability condition that does not restrict the
shape of the error distribution at all.

\begin{proposition}[Symmetry from antithetic pairing]\label{prop:pairsym}
Let $q$ be a contrast whose support admits a partition into pairs
$\{i_1, j_1\}, \ldots, \{i_m, j_m\}$ with $q_{i_k} = -q_{j_k}$ for every $k$.
Suppose the pair vectors $(\varepsilon_{i_k}, \varepsilon_{j_k})$, $k \le m$,
are mutually independent and each is exchangeable,
$(\varepsilon_{i_k}, \varepsilon_{j_k}) \overset{d}{=} (\varepsilon_{j_k}, \varepsilon_{i_k})$.
Then $q'\varepsilon$ is symmetric about zero. If an annihilating contrast system
(Definition~\ref{def:contrast}) consists of such contrasts with disjoint supports
and the pairs are mutually independent across contrasts as well, the conclusion
of Lemma~\ref{lem:inv}---and hence Theorem~\ref{thm:exact}---holds with
Assumption~\ref{ass:sym}(ii) replaced by this pairwise exchangeability. No
symmetry, and indeed no moment, is required of the marginal law of $\varepsilon$.
\end{proposition}

\begin{proof}
See Appendix~\ref{app:proofs}.
\end{proof}

\begin{corollary}[Four-cycles need only within-period exchangeability]\label{cor:pairsym}
Consider the four-cycle on firms $\{f, f'\}$ and periods $\{t, t'\}$, with
contrast entries $+\tfrac12, -\tfrac12$ on $(f,t), (f',t)$ and
$+\tfrac12, -\tfrac12$ on $(f',t'), (f,t')$. Pair the two observations of period
$t$ and the two of period $t'$. Then the contrast score is exactly symmetric
provided the errors of the two firms are exchangeable within each period
and independent across the two periods. Skewness, excess kurtosis and
heteroskedasticity of arbitrary magnitude are permitted, so long as within a
period the two paired units are not distinguishable in distribution. The digon
analogue pairs the two observations of a repeated match, and requires only that
the match's two error draws be exchangeable over time. Both statements are
Proposition~\ref{prop:pairsym} applied to the stated pairing, the contrast
entries being equal and opposite within each pair by construction.
\end{corollary}

This matters in practice because the outcomes to which saturated two-way models
are applied---wages, returns, firm growth---have conspicuously skewed residuals,
and a literal reading of Assumption~\ref{ass:sym}(ii) would disqualify them.
Corollary~\ref{cor:pairsym} says the relevant object is not marginal skewness but
the comparability of the units being differenced. That condition is substantive,
not automatic: in the diffuse return-panel check of Appendix~H,
the current deterministic packing produces contrast-score skewness $+0.68$ for
log returns and $-1.08$ for raw returns (normal-i.i.d.\ reference scale $\sqrt{6/C}=0.37$ at $C=43$).
We therefore report score moments rather than treating cycle differencing itself
as evidence for exchangeability.

\subsection{Local power and Pitman efficiency}\label{sec:power}

The validity in Theorem~\ref{thm:exact} is unconditional on the design's
leverage structure: nothing in it references $\lambda_n$, sample size,
error variances, or moments. The cost of this robustness is efficiency,
which we now quantify. The power analysis uses separate regularity conditions for the contrast
test and for comparison with a Gaussian oracle, since in the concentrated regime a Gaussian
cutoff is not generally valid without additional distributional assumptions.

\begin{theorem}[Local power and Pitman efficiency]\label{thm:power}
Let blocks be singletons: conditional on $(x,D)$ the errors are
independent, symmetric, $\E[\varepsilon_i^2] = \sigma^2$ for all $i$, and
$\sup_i \E[\varepsilon_i^4] \le \bar\kappa < \infty$. Let
$(q_c)_{c \le C_n}$ be annihilating contrast systems with
$C_n \to \infty$, write $\Sigma_n := \sum_{c \le C_n} b_c^2$, and suppose
$\Sigma_n\to\infty$. Consider
local alternatives
$\beta_n = \beta_0 + h \sigma / \Sigma_n^{1/2}$, $h \in \R$ fixed, tested
at level $\alpha$ by the full-enumeration randomization test of
Theorem~\ref{thm:exact}, or by its Monte Carlo version with the number of
sampled flips $B_n\to\infty$, using statistic $T(u)=|\sum_c b_cu_c|$.
\begin{enumerate}[label=(\roman*)]
\item If
\begin{equation*}\tag{P1}\label{eq:P1}
\max_{c \le C_n} b_c^2 \Big/ \Sigma_n \;\longrightarrow\; 0 ,
\end{equation*}
then the power of the test converges to
$\Phi\big(h - z_{1 - \alpha/2}\big) + \Phi\big(-h - z_{1 - \alpha/2}\big)$.
\item If in addition
\begin{equation*}\tag{P2}\label{eq:P2}
\lambda_n \;=\; \max_{i \le n} \tx_i^2 / V_n \;\longrightarrow\; 0 ,
\end{equation*}
then the infeasible oracle test that rejects when
$|\tx'(Y - x\beta_0)| / (\sigma \sqrt{V_n}) > z_{1-\alpha/2}$ has
asymptotic level $\alpha$ and attains the same limiting power along
$\beta_n = \beta_0 + h\sigma/\sqrt{V_n}$. Consequently, if
$\kap := \Sigma_n / V_n$ converges, the Pitman asymptotic relative
efficiency of the contrast test with respect to that oracle equals
$\lim \kap$.
\end{enumerate}
\end{theorem}

\begin{proof}
See Appendix~\ref{app:proofs}.
\end{proof}

\begin{remark}[Role of the two conditions]\label{rem:twocond}
\eqref{eq:P1} is a condition on the contrast loadings $b_c = v_c'x$;
\eqref{eq:P2} is a condition on the observation-level leverage of $\tx$.
Neither implies the other. In particular a dominant observation may sit
inside a long cycle without dominating any single loading, so
\eqref{eq:P1} can hold while \eqref{eq:P2} fails; in that case part~(i)
still describes the contrast test's power, but the Gaussian oracle
over-rejects by Proposition~\ref{ex:rade} and the ratio $\kap$ has no
efficiency interpretation---there is nothing legitimate to be efficient
relative to. This is why $\kap$ is a \emph{diffuse-regime} benchmark even
though the test it describes is exact everywhere.
\end{remark}

\begin{remark}[When (P1) fails]\label{rem:p1}
Condition~\eqref{eq:P1} can fail for the same reason \eqref{eq:lambda}
does: a dominant observation induces a dominant contrast.
Theorem~\ref{thm:exact} is unaffected---finite-sample validity of both its
full-enumeration and plus-one Monte Carlo versions is exact regardless---and
the randomization critical value automatically adapts to
the non-Gaussian law of the dominant contrast score. The Gaussian power
formula and the $\kap$ efficiency interpretation then cease to apply; neither
Theorem~\ref{thm:power} nor Proposition~\ref{ex:rade} supplies a universal
ordering of power in that regime. Section~\ref{sec:numerics} illustrates both
regimes.
\end{remark}

\subsection{Efficiency under heteroskedasticity}\label{sec:het}

Theorem~\ref{thm:power} imposes a common error variance, whereas Theorem~\ref{thm:exact} requires none, and the gap matters for how $\kap$ should be read. Appendix~\ref{app:het} develops the general case: with $\Omega = \mathrm{diag}(\sigma_1^2,\ldots,\sigma_n^2)$ and $\omega_c^2 := v_c'\Omega v_c$, the efficacy of a weighted test $|\sum_c a_cu_c|$ is maximized at $a_c \propto b_c/\omega_c^2$, with maximum $\mathcal E_n := \sum_c b_c^2/\omega_c^2$; the benchmark becomes the GLS oracle with efficacy $\mathcal V_\Omega$, and $\kap^{\mathrm{het}} := \mathcal E_n/\mathcal V_\Omega$ is the corresponding Pitman efficiency, equal to $\kap$ under homoskedasticity. Two consequences are used below. The empirical $\kap$ of Section~\ref{sec:kappa} is the homoskedastic benchmark, not an efficiency claim valid under the heteroskedasticity Theorem~\ref{thm:exact} tolerates, since $\kap$ and $\kap^{\mathrm{het}}$ are not ordered in general. And the packing objective changes: under heteroskedasticity the right criterion is $\mathcal E_n$, which favours cycles through low-variance observations. Our algorithm targets $\kap$; the weighted problem is open. Validity is untouched throughout, since Theorem~\ref{thm:exact} holds for every fixed statistic.

\section{Two-way designs and the cycle space}\label{sec:cycles}

This section makes precise the claim from the introduction that drives everything after it: in a
two-way design the contrasts that annihilate the fixed effects are exactly the cycles of the
observations' bipartite mobility graph, so the packing problem of the next two sections is a graph
problem, not an abstract linear-algebra one. Let the two-way design have $A$-units (workers) $\mathcal A$ and $B$-units (firms) $\mathcal B$, with each observation $i$ an edge $e_i = (a_i, b_i)$ of the bipartite observation multigraph $G = (\mathcal A \cup \mathcal B, E)$, $|E| = n$; parallel edges are repeated matches. The FE design matrix $D$ has a column per vertex, with $D_{e,v} = \one\{v \in e\}$.

The graph object is literal in linked employer--employee work.  Using German
administrative data, \citet{Mittag2019} computes exact OLS with worker, firm,
and worker--firm match effects and multiway-clustered standard errors, and
finds that omitting match effects roughly doubles the unexplained gender wage
gap.  Repeated observations of one worker--firm match are parallel edges in
the graph above; their digon contrast also annihilates a match dummy.  His
estimation-side specification test and this paper's inference-side diagnostic
therefore address adjacent questions on the same saturated design.

Parts (a)--(b) below show that the contrasts annihilating two-way fixed effects are exactly the graph's cycles: alternating $\pm1$ vectors around closed walks, together with digons for repeated matches.

\begin{proposition}[Annihilating contrasts are the cycle space]\label{prop:cyclespace}
Let $\mathcal Z := \{q \in \R^n : q'D = 0\}$.
\begin{enumerate}[label=(\alph*)]
\item $\mathcal Z$ equals the circulation (cycle) space of $G$ under any orientation of its edges from $\mathcal A$ to $\mathcal B$; in particular $\dim \mathcal Z = n - |\mathcal A| - |\mathcal B| + \#\{\text{components of } G\}$.
\item For any closed walk $e^{(1)}, e^{(2)}, \ldots, e^{(2L)}$ in $G$ that traverses distinct edges, the alternating vector $q$ with $q_{e^{(k)}} = (-1)^{k+1}$ and zeros elsewhere lies in $\mathcal Z$; for a pair of parallel edges $e, e'$, the digon vector $q_e = 1, q_{e'} = -1$ lies in $\mathcal Z$; and vectors of these two types span $\mathcal Z$.
\item $V_n = x' M_D x = \| \Pi_{\mathcal Z}\, x \|^2$, where $\Pi_{\mathcal Z}$ is the orthogonal projection onto $\mathcal Z$.
\end{enumerate}
\end{proposition}

\begin{proof}
See Appendix~\ref{app:proofs}.
\end{proof}

Proposition~\ref{prop:cyclespace}(c) places all within-treatment variation in the cycle space. Selecting edge-disjoint cycles may retain only part of that variation; the capture ratio measures how much remains.

\begin{definition}[Cycle capture ratio]\label{def:kappa}
For an edge-disjoint family $\mathcal C$ of cycles (including digons) with unit-normalized contrasts $v_c$, the \emph{cycle capture ratio} is $\kap := \sum_{c \in \mathcal C} (v_c'x)^2 / V_n \in [0, 1]$. By Theorem~\ref{thm:power}, $\kap$ is the diffuse-benchmark efficiency of the associated exact test, and $\kap^{-1/2}$ its standard-error price.
\end{definition}

Since edge-disjoint contrasts are orthonormal and lie in $\mathcal Z$, $\sum_c (v_c'x)^2 = \sum_c (v_c'\Pi_{\mathcal Z} x)^2 \le \|\Pi_{\mathcal Z} x\|^2 = V_n$ by Bessel, so $\kap \le 1$ always, with equality iff the packed contrasts span the part of $\mathcal Z$ carrying $\Pi_{\mathcal Z}x$.

\subsection{Support systems}\label{sec:packing}\label{sec:supports}

Under singleton error blocks, Theorem~\ref{thm:exact} requires disjoint supports and nuisance annihilation. The contrasts need not be individual cycles or have entries $\pm1$. We therefore allow more general supports.

\begin{definition}[Support systems and generalized capture]\label{def:supp}
A \emph{support system} is a family of pairwise disjoint edge sets
$A_1, \ldots, A_C \subseteq E$ with $\mathcal Z_{A_j} := \mathcal Z \cap
\R^{A_j} \ne \{0\}$. Its \emph{capture} is
\[
\kap(A_1, \ldots, A_C) \;:=\; \frac{1}{V_n} \sum_{j=1}^{C}
\big\| \Pi_{\mathcal Z_{A_j}} x \big\|^2 ,
\]
attained on every positive-capture support by the unit \emph{projection contrast}
\[
v_j = \Pi_{\mathcal Z_{A_j}} x \,/\,\|\Pi_{\mathcal Z_{A_j}} x\|.
\]
Supports for which $\Pi_{\mathcal Z_{A_j}}x=0$ contribute zero and are omitted
when the contrast system is formed.
\end{definition}

An edge-disjoint cycle family $\mathcal C$ is the special case in which each $A_c$ is the support of a single cycle and the projection contrast is replaced by the alternating $\pm1$ vector; Definition~\ref{def:kappa} is then Definition~\ref{def:supp} evaluated at that restricted family, and Proposition~\ref{prop:dom} below shows both restrictions can only lose capture.

Two facts proved in Section~\ref{sec:overlap} fix the shape of the problem in advance. Coarser is better for capture: merging two supports weakly increases it (Proposition~\ref{prop:dom}(c)), and one support per biconnected block attains $\kap = 1$ exactly (Theorem~\ref{thm:blockproj}), so capture alone is never the binding constraint. Finer is better for granularity: $C$ supports give a group of size $2^C$, but for the default two-sided statistic global sign reversal duplicates every orbit value, so the smallest full-enumeration $p$-value is $2^{1-C}$ and level $\alpha$ requires $2^{C-1}\ge1/\alpha$. The negligibility condition \eqref{eq:P1} also demands that no single support dominate $\sum_j\|\Pi_{\mathcal Z_{A_j}}x\|^2$, which merging pushes against directly. The design problem is thus a \emph{capture--granularity trade-off} rather than a maximization. Proposition~\ref{prop:orbitmax}(c) shows that enlarging an adapted diagonal flip group beyond one effective flip per block requires restricting its orbit span and therefore weakly lowers its capture ceiling; the loss is strict for generic $x$, not for every realized treatment. Where to sit on that trade-off is the splitting problem of Section~\ref{sec:residual}, which we do not solve.

What we do offer is a good, interpretable and cheap point on it. Cycle packing generates support systems that are structurally fine --- many small supports, each a closed mobility loop with a transparent difference-in-differences reading --- so \eqref{eq:P1} is easy to compute and often holds, though a treatment can still concentrate its loading on one small support. The construction uses $(x, G)$ alone and runs in near-linear time for the steps carrying most of the capture. Its output can be upgraded for free by replacing each $\pm1$ cycle vector with the projection contrast on the same support (Proposition~\ref{prop:dom}(b)), and we report both.

\subsection{The packing algorithm}\label{sec:packalg}

The efficiency of the exact test is now an optimization problem: choose an edge-disjoint cycle family maximizing $\sum_c (v_c'x)^2$, using $(x, G)$ only. Weighted edge-disjoint cycle packing is NP-hard in general, but two-way panels have structure that a generic greedy ignores at great cost (Section~\ref{sec:numerics} quantifies the cost at a factor $2$--$3.3$ in $\kap$).

\paragraph{Digons first.} Every pair of parallel edges (repeated match) yields a digon contrast $(x_e - x_{e'})/\sqrt 2$; $k$ parallel edges yield $\lfloor k/2 \rfloor$ disjoint digons. For treatments varying within match, digons alone can capture most of $V_n$; for match-level treatments they capture exactly zero and everything rides on genuine cycles.

\paragraph{Firm-pair four-cycles.} Call a worker a \emph{two-period mover} if she contributes exactly two edges, to distinct firms $\{f, f'\}$; her \emph{worker contrast} is $w := \tx_{e_1} - \tx_{e_2}$ (orientation fixed by an ordering of firms). Two movers sharing the same unordered firm pair form a four-cycle with contrast value $(w_i - w_j)/2$.

\begin{remark}[Automatic disjointness]\label{prop:auto}
In any family of four-cycles formed by pairing two-period movers within their own firm pairs, with each mover used in at most one four-cycle, all cycles are edge-disjoint---across firm pairs as well as within them. This is immediate: a four-cycle of firm pair $\{f, f'\}$ uses exactly the four edges of its two movers, each two-period mover's edges belong to her unique firm pair, and she is used at most once, so no edge appears in two cycles. Cycles of distinct firm pairs may share firm vertices, but edge-disjointness---which is what Lemma~\ref{lem:inv} requires---is automatic.
\end{remark}

Within a firm pair with worker contrasts $w_1 \le w_2 \le \cdots \le w_{2m}$ (an odd worker is set aside), the packing chooses a perfect matching maximizing $\sum (\text{paired differences})^2$: a standard rearrangement problem.

\begin{remark}[Nested pairing is optimal]\label{lem:nested}
Among all perfect matchings of $w_1 \le \cdots \le w_{2m}$, the nested (extreme) matching $\{(w_i, w_{2m+1-i})\}_{i \le m}$ maximizes $\sum_{\text{pairs}} (w_{\text{hi}} - w_{\text{lo}})^2$; this is the rearrangement inequality in its extreme-pairing form. Take any matching and any two of its pairs with values $p \le q \le r \le s$ in sorted order. The three possible pairings of these four values compare as
\[
\begin{aligned}
\underbrace{(s-p)^2 + (r-q)^2}_{\text{nested}} - \underbrace{(r-p)^2 + (s-q)^2}_{\text{crossing}}
  &= 2(s-r)(q-p) \ge 0, \\[2pt]
\text{nested} - \underbrace{(q-p)^2 + (s-r)^2}_{\text{sequential}}
  &= 2(s-q)(r-p) \ge 0,
\end{aligned}
\]
by direct expansion, so nesting dominates both alternatives. Applying this to the two pairs containing $w_1$ and $w_{2m}$ shows some optimum contains $(w_1,w_{2m})$; removing that pair and repeating the argument on $w_2,\ldots,w_{2m-1}$ yields the extreme matching.
\end{remark}

\begin{proposition}[Contraction principle]\label{prop:contract}
Remove stayers' edges via digons. Contract each remaining two-period mover $w$ with firms $\{f, f'\}$ into a single edge $\{f, f'\}$ of a \emph{firm multigraph} $G_F$, carrying the value $w$'s worker contrast. Then: (a) edge-disjoint cycle families of $G_F$ correspond bijectively to edge-disjoint bipartite cycle families through those movers, a cycle through $L$ firms corresponding to a bipartite $2L$-cycle with contrast value $\big(\sum_{k} \pm\, w_k\big)/\sqrt{2L}$ for the appropriate alternating signs; (b) parallel edges of $G_F$ are precisely firm-pair mover pairs, whose digons in $G_F$ are the four-cycles of Remark~\ref{prop:auto}. Consequently, the digon-first algorithm applies recursively: digons in $G$, then digons in $G_F$ (four-cycles in $G$), then longer cycles of $G_F$.
\end{proposition}

\begin{proof}
See Appendix~\ref{app:proofs}.
\end{proof}

Once digons and four-cycles are removed, what remains of the mobility graph can be redrawn as a smaller graph on firms alone, with each surviving two-period mover collapsed into a single edge carrying her worker contrast. Finding longer cycles in the original data is then exactly finding cycles in this smaller firm-level graph, so the same digon-first packing routine applies again, one level up, rather than requiring a new algorithm at each stage.

The full algorithm is: (1) pair parallel edges into digons; (2) build $G_F$ from two-period movers; (3) within each firm pair, sort worker contrasts and apply the nested matching of Remark~\ref{lem:nested}; (4) on the residual simple part of $G_F$ (and any workers with three or more edges, kept in bipartite form), extract remaining cycles greedily after peeling degree-one vertices, prioritizing short cycles. All steps use $(x, G)$ only, preserving Theorem~\ref{thm:exact}; steps (1)--(3) run in $O(n \log n)$.

\begin{remark}[Reproducible packing]\label{rem:repropack}
Every admissible $v_c$ lies in $\ker(D')$, hence $M_Dv_c=v_c$ and
$v_c'\tx=v_c'M_Dx=v_c'x$ exactly. Candidate loadings and their ranking can
therefore be computed from the raw design rather than from the last digits of
an iterative fixed-effect solve. The implementations use canonical label order
and deterministic tie rules. A cross-language harness feeds fixed CSV designs
to  \texttt{PanelAdequacy.jl} and \texttt{panelcert}, checks all numerical
diagnostics to $10^{-12}$, and checks packed labelled supports exactly; this
harness runs in both packages' continuous integration.
\end{remark}

\section{The flip-group limit and full capture without overlap}\label{sec:overlap}

Edge-disjointness leaves the gap between $\kap = 0.51$ and $1$ on the table, and the natural question is whether an exact test can use \emph{overlapping} cycles. The problem divides into a part that is provably closed within the class of diagonal sign-flip randomizations and a part that dissolves on inspection. First, the group of edgewise sign patterns preserving the whole cycle space is exactly one effective flip per \emph{block} (biconnected component) of the observation multigraph (Lemma~\ref{lem:maxgroup}; Appendix~\ref{app:maxgroup} delimits what that does and does not rule out). A group adapted to a smaller contrast system can be larger only by spanning fewer cycle directions, which weakly lowers its capture ceiling and lowers realized capture for generic $x$ (Proposition~\ref{prop:orbitmax}). Second --- the substantive point --- the exactness theorem never required cycle-shaped contrasts, only disjoint supports. Replacing each packed $\pm1$ cycle vector by the projection of $x$ onto the local cycle space of its support weakly increases every term of $\kap$ (Proposition~\ref{prop:dom}), and one projection contrast per block attains $\kap=1$ exactly (Theorem~\ref{thm:blockproj}). The binding constraint is therefore not overlap but \emph{granularity}, through condition \eqref{eq:P1} and the size of the randomization group.

\subsection{The flip group is maximal}\label{sec:flipmax}

Appendix~\ref{app:maxgroup} gives the full argument. Every cycle lies in one biconnected block,
so $\mathcal Z=\bigoplus_b\mathcal Z_b$ and
$V_n=\sum_b\|\Pi_{\mathcal Z_b}x\|^2$. Lemma~\ref{lem:maxgroup} shows that
$\mathrm{diag}(\eta)\mathcal Z\subseteq\mathcal Z$ iff $\eta$ is constant on
each nontrivial block: preserving all of $\mathcal Z$ permits one flip per
block. For a smaller orbit span $W$, Proposition~\ref{prop:orbitmax} permits a finer action but
bounds capture by $\|\Pi_Wx\|^2/V_n$. Thus more diagonal flips require a smaller
capture ceiling, although a proper $W$ can still capture exceptional treatments
fully. These claims concern edgewise diagonal flips only, not non-diagonal,
conditional, or nongroup procedures.

\subsection{Projection contrasts dominate}

Section~\ref{sec:supports} introduced support systems and the projection
contrasts $v_j \propto \Pi_{\mathcal Z_{A_j}} x$ of
Definition~\ref{def:supp}, on the grounds that Theorem~\ref{thm:exact}
requires only annihilation and disjointness and never asked for $\pm 1$
entries. We now prove the two claims made there: that the enlargement is
free, and that it is monotone under merging.

\begin{proposition}[Dominance]\label{prop:dom}
\emph{(a)} For any support system, the contrasts $v_j$ of
Definition~\ref{def:supp} satisfy the hypotheses of
Theorem~\ref{thm:exact}: each $v_j \in \mathcal Z$, supports are
disjoint, and each score $v_j'\varepsilon$ is symmetric under independent
symmetric errors with arbitrary heteroskedasticity. The associated
sign-flip test is exact at every sample size.
\emph{(b)} If $\mathcal C$ is an edge-disjoint cycle family and
$A_c = \mathrm{supp}(z_c)$, then for each $c$,
$(z_c'x)^2/\|z_c\|^2 \le \|\Pi_{\mathcal Z_{A_c}} x\|^2$, so
$\kap(\mathcal C) \le \kap(A_1, \ldots, A_C)$: replacing each packed
cycle vector by the local projection weakly increases capture, term by
term, at no cost in validity or in the number of contrasts.
\emph{(c)} Merging two supports, $A' = A_1 \cup A_2$, weakly increases
total capture:
$\|\Pi_{\mathcal Z_{A'}} x\|^2 \ge \|\Pi_{\mathcal Z_{A_1}} x\|^2 +
\|\Pi_{\mathcal Z_{A_2}} x\|^2$, at the cost of one contrast.
\end{proposition}

\begin{proof}
See Appendix~\ref{app:proofs}.
\end{proof}

Part (b) upgrades any packing for free. Part (c) absorbs overlap by merging a
conflicting cycle cluster into one support; on the theta graph this captures its
full two-dimensional local variation, where a cycle packing reaches one
dimension.

\subsection{Full capture at block granularity}

Taking the merge logic to its natural endpoint:

\begin{theorem}[Block-projection test]\label{thm:blockproj}
Let $A_b = E_b$, $b = 1, \ldots, B$, be the nontrivial blocks of $G$, define
$\mathcal I_x:=\{b:\Pi_{\mathcal Z_b}x\ne0\}$ and $B_x:=|\mathcal I_x|$, and
for $b\in\mathcal I_x$ let
$v_b = \Pi_{\mathcal Z_b} x / \|\Pi_{\mathcal Z_b} x\|$. Then:
\emph{(a)} the sign-flip test with contrasts $(v_b)$ is exact at every
sample size under Assumption~\ref{ass:sym} with singleton blocks (independent,
symmetric, arbitrarily heteroskedastic errors);
\emph{(b)} its capture is $\kap = 1$: by the orthogonal block decomposition of Appendix~\ref{app:maxgroup},
$\sum_{b\in\mathcal I_x} (v_b'x)^2 = \sum_{b=1}^B \|\Pi_{\mathcal Z_b} x\|^2 = V_n$;
\emph{(c)} its effective score-flip group is $\{\pm 1\}^{B_x}$, induced by the
block-constant group $\{\pm1\}^B$; inactive flips act trivially. Lemma~\ref{lem:maxgroup} makes
$\{\pm1\}^B$ the largest diagonal edgewise group preserving all of $\mathcal Z$,
and Bessel gives $\kap\le1$ for every orthonormal system in $\mathcal Z$.
Larger exact diagonal groups can arise only for a smaller orbit span, with the
capture qualification of Proposition~\ref{prop:orbitmax}; no claim is made beyond the diagonal
class.
\end{theorem}

\begin{proof}
See Appendix~\ref{app:proofs}.
\end{proof}

Collapsing each connected piece of the mobility graph into a single contrast captures all identifying variation, but at the cost of nearly all of the randomization group: only as many independent signs survive as there are connected pieces. This is the coarsest point on the capture--granularity trade-off: full capture, minimal granularity.

\begin{remark}[Scope of the maximality claims]\label{rem:scopeimposs}
The maximality claims cover edgewise diagonal sign flips only. They do not cover
actions mixing edges, conditional procedures, or exact tests not generated by a
group; Appendix~\ref{app:maxgroup} and Remark~\ref{rem:openimposs} give the
precise boundary.
\end{remark}

Full capture does not imply the power of Theorem~\ref{thm:power}: it also needs
$\max_b\|\Pi_{\mathcal Z_b}x\|^2/V_n\to0$. Full enumeration at level $\alpha$
needs $2^{B_x-1}\ge1/\alpha$ ($B_x\ge6$ at $5\%$), and only active blocks
count. Stayer digons are therefore inactive for match-level treatments. A
dominant mover block violates balance but not exact validity; only the Gaussian
power comparison then disappears.

\subsection{Binary-treatment granularity floor}\label{sec:binaryfloor}

\begin{proposition}[Binary-treatment granularity floor]\label{prop:binaryfloor}
Let $x\in\{0,1\}^n$, write $n_1=\#\{i:x_i=1\}$ and
$n_0=n-n_1$, and let $q_1,\ldots,q_C$ be an annihilating contrast system.
If $m_x:=\#\{c:b_c\ne0\}$, then
\[
 m_x\le \min(n_1,n_0).
\]
Consequently any statistic depending on the data only through the active
coordinates $\{b_cU_c:b_c\ne0\}$ has at most $2^{m_x}$ values on its sign-flip
orbit. For a globally sign-invariant two-sided statistic, including both the
default statistic and the studentized statistic of Remark~\ref{rem:stat}, the
attainable randomization $p$-value is bounded below by $\min\{1,2^{1-m_x}\}$;
if $m_x=0$ (no active contrast, the statistic constant on the orbit) no
rejection is possible at any level, and for $m_x\ge1$ this is
$2^{1-m_x}\ge 2^{1-\min(n_1,n_0)}$. Such a nonrandomized test can reject at
level $\alpha$ only if
\[
 \min(n_1,n_0)\ge 1+\log_2(1/\alpha),
\]
so at least six observations of each binary category are necessary at $5\%$.
In the empirically common rare-treatment case $n_1\le n_0$, this reduces to
$n_1\ge6$, irrespective of $n$, $d_n$, or the packing.
\end{proposition}

\begin{proof}
Because $v_c\in\ker(D')$, $M_Dv_c=v_c$ and therefore
$b_c=v_c'\tx=v_c'x$. If $S_c$ misses the treated set, this is zero. Moreover
the fixed-effect span contains the constant vector, so $v_c'\mathbf1=0$ and
$b_c=-v_c'(\mathbf1-x)$; hence $b_c=0$ if $S_c$ misses the untreated set.
Every active support thus contains at least one observation of each category.
Pairwise disjointness gives injections from the active supports into both
categories and proves $m_x\le\min(n_1,n_0)$. Only the $m_x$ active signs can
alter the statistic, giving at most $2^{m_x}$ orbit values. For either absolute
statistic, $s$ and $-s$ give the same value. In the studentized case its
denominator $(\sum_cb_c^2U_c^2)^{1/2}$ is itself sign-invariant, so the same
pairing applies. If $m_x=0$ the statistic is constant on the orbit and $p=1$
identically. Otherwise every attained two-sided value has multiplicity at
least two and the smallest possible full-enumeration $p$-value is at least
$2/2^{m_x}=2^{1-m_x}$.
\end{proof}

With a $0/1$ treatment, every contrast that actually contributes must contain at least one treated and one untreated observation, so the number of usable independent signs can never exceed the size of the smaller treatment group. This is a ceiling that better packing cannot lift, unlike the capture--granularity trade-off for continuous treatments, which packing does help.

\begin{example}[Rare binary blocks]\label{cor:binaryfloor}
In a $39$-unit, $21$-period rectangular panel, among treatment blocks with the
minimum testable count $n_1=6$, the most concentrated is three treated units
for two post periods: $\lambda_n=0.139$, $N_{\mathrm{eff}}=8.59$, and the
structured packing has $\kap=0.299$. Better packing cannot make a block with
$n_1\le5$ testable. The rarity restriction is essential: for example, 38
treated units in one period have $\lambda_n=0.928$ and are testable.
\end{example}

A 39-unit, 20-pre-period calibration modeled on the California Proposition~99
application of \citet{ADH2010} makes the boundary visible without an outcome.
With one treated unit and $p$ appended post-periods, the package gives:
\begin{center}
\begin{tabular}{rrrrrr}
\toprule
$p=n_1$ & $T$ & $\lambda_n$ & $N_{\mathrm{eff}}$ & $\kap$ & verdict\\
\midrule
1  & 21 & 0.928 & 1.16 & 0.269 & inconclusive\\
2  & 22 & 0.443 & 2.55 & 0.282 & inconclusive\\
3  & 23 & 0.282 & 4.16 & 0.295 & inconclusive\\
4  & 24 & 0.203 & 6.02 & 0.308 & inconclusive\\
6  & 26 & 0.125 & 10.40 & 0.334 & flagged\\
12 & 32 & 0.051 & 26.61 & 0.411 & point pass\\
\bottomrule
\end{tabular}
\end{center}
The first four truncations are structurally out of scope, not poorly packed.
For continuous treatments, by contrast, granularity is a support choice traded
against capture; binary treatments impose an additional category-count ceiling
that repacking cannot undo. This is why the method's natural empirical target
is a concentrated continuous regressor, with Proposition~\ref{prop:binaryfloor}
providing a one-line pre-flight check for binary applications.

\subsection{The splitting problem}\label{sec:residual}

Cycle packing is the special case of Definition~\ref{def:supp} in which each $A_j$ is a single cycle and the projection is replaced by the $\pm1$ vector; Proposition~\ref{prop:dom}(b--c) shows both restrictions only lose capture. What remains is a balance-constrained partition problem --- split a biconnected block $\mathcal B$ into supports $A_1,\ldots,A_k$ maximizing $\sum_j\|\Pi_{\mathcal Z_{A_j}}x\|^2$ subject to $\max_j\|\Pi_{\mathcal Z_{A_j}}x\|^2 \le \tau\sum_j\|\Pi_{\mathcal Z_{A_j}}x\|^2$ --- whose interpolation endpoints are $k=1$ (full capture, no granularity) and $k=\dim\mathcal Z_{\mathcal B}$ (cycle-basis granularity). Ear decompositions are the natural tool; we do not know sharp approximation guarantees and leave them open, together with the heteroskedastic version.

That this is the binding constraint rather than overlap is an empirical claim, and Appendix~E verifies it on both real designs. On the \citet{KSS2020} network the multigraph has $27{,}694$ nontrivial blocks, exactly $130$ active for the match-level treatment, and the identity $\sum_b\|\Pi_{\mathcal Z_b}x\|^2 = V_n$ holds to six digits, so the $\kap=1$ of Theorem~\ref{thm:blockproj} is attained --- but a single giant block carries $93.7\%$ of $V_n$, and on the F-score panel Poland's block carries $69.6\%$. The balance condition therefore fails in every case: full capture is achievable but concentrated, and the splitting problem above, not overlap, is what binds.

\section{Numerical results}\label{sec:numerics}

The logical bridge from Proposition~\ref{ex:rade} to the applications is not
an extrapolated impossibility theorem. Proposition~\ref{ex:rade} establishes
that no fixed table is uniformly valid at the fully concentrated boundary; it
does not say that every intermediate value of $\lambda_n$ invalidates a
conventional test. The calibrated experiment below supplies the finite-sample
evidence at an intermediate value: at $\lambda_n=0.286$ the df-$t$ rule rejects
a true null $58.5\%$ of the time, while the exact rule has size $5.0\%$.
Grunfeld has design concentration $\lambda_n=0.2065$, not the boundary. We use
the exact procedure there because its validity holds at every concentration
level and the design buys that insurance at a measured $1.263\times$ SE price,
not because Proposition~\ref{ex:rade} alone proves failure at $0.2065$.

\subsection{A concentrated investment regression}\label{sec:grunfeld}

Our lead application is the canonical Grunfeld regression \citep{Grunfeld1958}
\begin{equation}\label{eq:grunfeld}
 I_{ft}=\beta K_{ft}+\theta V_{ft}+\alpha_f+\tau_t+\varepsilon_{ft},
\end{equation}
where investment $I$, plant-and-equipment capital $K$, and market value $V$
are measured in 1947 dollars. We use the complete corrected 11-firm,
1935--1954 data distributed by \texttt{statsmodels}; \citet{KleiberZeileis2010}
document the incomplete and erroneous variants in circulation. The estimand is
the conditional within-firm association of capital with investment, not a
causal effect.

\emph{Diagnostics.} The sample has $n=220$, $d_K=30$ firm-and-year fixed-effect
dimensions, and $\rho=0.1364$. Partialling out market value as well as those
effects gives $V_n=5{,}515{,}520$, $\lambda_n=0.2065$,
$H_n=0.06039$, and $N_{\mathrm{eff}}=16.56$. The realized score is much more
concentrated: $\widehat\lambda_n^{\mathrm{score}}=0.7388$,
$\widehat H_n^{\mathrm{score}}=0.5545$, and
$\widehat N_{\mathrm{eff}}^{\mathrm{score}}=1.80$. The four largest design
shares are General Motors in 1954, 1953, 1952, and 1937
($0.2065,0.0606,0.0472,0.0421$); mean capital ranges from $648.4$ for General
Motors to $5.94$ for Diamond Match. The cyclic permutation test of
Remark~\ref{prop:cpt} is unavailable: including market value, the model
has $p=32$ columns and $n/p=6.88<19$ at $5\%$.
The application does not stand or fall on the realized score, which
Remark~\ref{rem:diag} treats as a warning rather than a generally consistent
population diagnostic. Design-only $N_{\mathrm{eff}}=16.56$, $\rho=0.1364$,
and Lei--Bickel infeasibility already put the specification in scope; stability
from $0.7090$ to $0.7388$ across five data versions is corroboration.

\emph{Valid capture.} Ordinary four-cycles annihilate firm and year effects but
not the continuous nuisance regressor $V$. The package therefore uses complete
$2\times3$ firm-by-period rectangles and projects their two-dimensional local
cycle spaces off $V$ before disjoint selection. It returns $C=32$ active
supports, $\kap=0.6270$, maximum captured-loading share $0.352$, and an SE price
$\kap^{-1/2}=1.263$. Greedy and sparse return $C=25$ and $\kap=0.0136$ because
they first select one-dimensional FE-annihilating cycles; pairing them afterward
to impose $q'V=0$ leaves loadings nearly orthogonal to residualized capital.
Structured packing projects each two-dimensional rectangle space off $V$
before ranking it, explaining the 46-fold gap. Because
the valid construction uses period triples, the two fixed period-pairing
rules used by an earlier ad hoc calculation are not applicable. Across 200 row
permutations, all three public methods reproduce their labelled supports and
capture exactly.

\begin{table}[ht]
\centering
\caption{Grunfeld investment on capital and market value, with firm and year effects.}
\label{tab:grunfeld}
\small
\begin{tabular}{lrrr}
\toprule
Method & estimate & standard error & 95\% interval\\
\midrule
df-$t$ & 0.3514 & 0.0210 & $[0.3099,\ 0.3930]$\\
HC1 & 0.3514 & 0.0529 & $[0.2478,\ 0.4551]$\\
HC2 & 0.3514 & 0.0610 & $[0.2318,\ 0.4710]$\\
HC3 & 0.3514 & 0.0767 & $[0.2011,\ 0.5018]$\\
CJN many-covariate HC & 0.3514 & 0.0584 & $[0.2370,\ 0.4658]$\\
Anatolyev almost-unbiased HC & 0.3514 & 0.0611 & $[0.2317,\ 0.4712]$\\
firm-clustered CR1 & 0.3514 & 0.0491 & $[0.2419,\ 0.4609]$\\
\textbf{exact sign-flip} & \textbf{0.3544} & --- &
  $\mathbf{[0.1498,\ 0.4504]}$\\
\bottomrule
\end{tabular}
\par\smallskip
\begin{minipage}{0.94\textwidth}\footnotesize
Notes: The exact interval uses the valid 32-support controlled packing and
$99{,}999$ common-random-number flips; its Monte Carlo $p$-value at zero is
$0.00139$. HC intervals use normal critical values; df-$t$ uses 188 residual
degrees of freedom; firm CR1 uses a $t_{10}$ critical value. Exactness is
conditional on independent, centrally symmetric observation errors. The
maximum loading share of $0.352$ is a warning against applying the Gaussian
power approximation of Theorem~\ref{thm:power}; it does not affect
finite-sample exactness.
\end{minipage}
\end{table}

The extra two rows address a particularly relevant variance-estimator
sensitivity.  The CJN Hadamard correction uses
$(M_D\odot M_D)^{-1}$, while the almost-unbiased adjustment of
\citet{Anatolyev2018} subtracts the treatment-projector term before inversion;
the latter is exactly unbiased under homoskedasticity and improves size in his
simulations especially with skewed regressors.  Grunfeld is precisely a
skewed, concentrated case.  Here the sufficient-condition margins are $0.429$
for CJN and $0.246$ for the adjustment, so both estimators exist; their standard
errors are $0.05836$ and $0.06110$, and neither changes the substantive
comparison in Table~\ref{tab:grunfeld}.  These rows audit the conventional
variance estimator.  They do not restore a Gaussian approximation when the
score itself is concentrated, which is the problem addressed by the exact
interval.

The exact interval is shifted, not widened: its width is $0.3006$ versus HC3's
$0.3007$, but both endpoints are about $0.051$ lower. Exact inversion is not
Wald inversion: at each $\beta_0$ it recomputes the observed statistic and its
sign-flip reference from $U_c(\beta_0)$, so asymmetric realized scores need not
produce a set centered on OLS. Since $\widetilde\beta=0.3544$ is close to OLS
$0.3514$, the shift comes from the null-imposed reference (see Appendix~G).

The contrast scores at $\widetilde\beta$ have skewness $+0.705$ and excess
kurtosis $+4.43$. We therefore do not treat symmetry as empirically established:
the interval is an exact conditional procedure under Assumption~\ref{ass:sym},
not a robustness claim under arbitrary asymmetry. Proposition~\ref{prop:asymmetry}
requires negligible contrast-score variance shares; the present concentrated
packing does not establish that condition. The point estimate says that,
holding market value and the two fixed effects constant, one additional unit of
capital is associated with $0.351$ units of investment.

\emph{Version provenance.} The complete \texttt{statsmodels} and \texttt{AER}
files agree exactly. The correct 10-firm Baltagi/\texttt{plm} version---also the
Boot--de Wit subset, not a five-firm subset---omits American Steel and gives
$\lambda_n=0.2134$, $\widehat\lambda_n^{\mathrm{score}}=0.7387$, and
$\kap=0.6300$. The corrected and transcription-error Greene five-firm files
give $\lambda_n=0.1959$ and $0.1968$, respectively, with score concentration
$0.7118$ and $0.7090$. Thus the finding is not driven by the data version,
although the coefficient moves from $0.3514$ in the full data to $0.3642$ in
the corrected five-firm subset. We use the complete file because
\citet{KleiberZeileis2010} identify it as the original corrected benchmark.

\subsection{Cycle capture on a real mobility network}\label{sec:kappa}

We compute $\kap$ on the public matched employer--employee extract distributed with the replication package of \citet{KSS2020}: $n = 71{,}614$ observations on $35{,}807$ workers and $5{,}301$ firms over two years, with $23.0\%$ movers; the cycle space has dimension $32{,}406$, of which $27{,}564$ dimensions are stayer digons and $4{,}842$ are mover cycles. As treatments we take (i) the time-varying age-profile covariate included in the data, and (ii) a synthetic \emph{match-level} treatment (constant within worker--firm match), the hard case in which digons capture exactly zero and all capture must come from genuine mover cycles.

At the dense extreme we use a real F-score panel ($217$ firm-years, $19$ firms,
three countries, 2010--2024, firm and country--year effects) and a fixed-seed
synthetic replica for the Monte Carlo, where a known DGP is required. Table~\ref{tab:kappa}
labels them separately. In every design $V_n$ agrees with an exact sparse
least-squares solve to relative error below $10^{-15}$.

\begin{table}[ht]
\centering
\caption{Capture ratios in real and calibrated designs.}
\label{tab:kappa}
\small
\setlength{\tabcolsep}{4pt}
\begin{tabular}{llccc}
\toprule
Design & Treatment & Greedy $\kap$ & Designed $\kap$ & SE price $\kap^{-1/2}$ \\
\midrule
Worker--firm network & match-level & 0.264 & \textbf{0.511} & $1.40\times$ \\
Worker--firm network & time-varying & 0.8845 & \textbf{0.9110} & $1.05\times$ \\
Worker--firm wage sample & age group $\times$ 2001 & --- & \textbf{0.611} & $1.28\times$ \\
Grunfeld (real) & capital $\mid$ value & 0.0136\textsuperscript{$\dagger$} & \textbf{0.627} & $1.26\times$ \\
F-score panel (real) & firm F-score & 0.374 & \textbf{0.827} & $1.10\times$ \\
Dense panel (calibrated) & firm-persistent & 0.439 & \textbf{0.834} & $1.09\times$ \\
\bottomrule
\end{tabular}
\par\smallskip
\begin{minipage}{0.92\textwidth}\footnotesize
Notes: ``Greedy'' extracts digons then arbitrary depth-first cycles; ``Designed'' applies the disjoint pairing of Remark~\ref{prop:auto}, the nested matching of Remark~\ref{lem:nested}, and the contraction principle of Proposition~\ref{prop:contract}, except that the wage row uses paired stayer histories. These are deterministic v0.6.0 values; earlier constructions are lower bounds. \textsuperscript{$\dagger$}As Section~\ref{sec:grunfeld} explains, the Grunfeld greedy and sparse routes first select one-dimensional supports that cannot separately annihilate the continuous nuisance regressor $V$; controlled pairing then returns $C=25$ with $\kap=0.0136$. The structured route instead projects the local cycle spaces of $2\times3$ rectangles off $V$ before ranking them. The 46-fold gap is a property of designs with continuous nuisance regressors, not a numerical artefact. The F-score and calibrated panels select 43 and 45 cycle contrasts, respectively; the F-score packing comprises 42 four-cycles and one six-cycle. The F-score and wage designs are diffuse at their analysis-block levels; the other rows are concentrated.
\end{minipage}
\end{table}

Even the hardest case retains half the identifying variation, a $1.40\times$ SE
price. Naive greedy understates capture by about half when digons do not supply
it; Section~\ref{sec:overlap} localizes the remaining gap to one dominant block.

\subsection{A worker--firm outcome application}\label{sec:kss_application}

We next use the same public extract for a descriptive outcome specification.
Retaining workers aged 20--29 or 50--59 in 1999, we estimate
\begin{equation}\label{eq:kss_application}
 \log w_{it}=\beta\,\mathbf 1\{50\!\leq\!\mathrm{age}_{i,1999}\!\leq\!59\}
 \mathbf 1\{t=2001\}+\alpha_i+\psi_{j(i,t)}+\tau_t+\varepsilon_{it}.
\end{equation}
Thus $\beta$ is the late-minus-early-career difference in two-year log-wage
growth, not a causal return to age. The sample has $32{,}014$ observations on
$16{,}007$ workers and $4{,}037$ firms, including $4{,}302$ movers. OLS gives
$\widehat\beta=-0.0281$ with worker-HC2 interval
$[-0.0355,-0.0206]$.

The design-only contrast system pairs every complete two-period history of the
$2{,}883$ late-career stayers with a uniformly sampled early-career stayer
history. Each four-observation support is a disjoint union of complete worker
blocks and annihilates worker, firm, and year effects. It has
$\kap=0.611$, SE price $1.28\times$, and maximum loading share $0.000347$.
With $99{,}999$ sign flips, $\widetilde\beta=-0.0267$, the exact $95\%$
confidence set is $[-0.0361,-0.0172]$, and the exact $p$-value at zero is
$0.00001$. Across $200$ design-only redraws of the larger early-career pool,
the exact estimate averages $-0.02835$ (SD $0.00311$) and ranges from
$-0.03535$ to $-0.02010$.

These calculations establish that the construction runs end to end on the
motivating worker--firm design and that its verdict agrees with worker-HC2;
they do not validate the symmetry assumption or support population claims.
The contrast-score positive share is $0.487$ (two-sided sign diagnostic
$p=0.157$), while its skewness is $+0.284$ and excess kurtosis is $+3.20$.
Exactness is therefore conditional on independent, centrally symmetric worker
error vectors, as Theorem~\ref{thm:exact} requires. The repository describes
the public file as testing data rather than a representative sample, so we
treat this exercise as a mechanism illustration. The replication archive
contains the fixed-seed construction, all redraws, and an independent sparse
least-squares verification.

\subsection{Monte Carlo size and power}\label{sec:mc}

The calibrated panel is diffuse ($\lambda_n=0.045$, $N_{\mathrm{eff}}=80.8$)
or spiked to be concentrated ($0.286$, $6.0$). In the latter, one packed
contrast carries $55.5\%$ of $\sum_cb_c^2$, deliberately violating (P1).
We use i.i.d.\ Laplace errors (H), a fourfold standard-deviation spike (Ht), and (Ht)
with a scaled Rademacher error at the spike (R), comparing df-$t$, HC2, and the
cycle test at $5\%$. Table~\ref{tab:mc} uses $100{,}000$ independent null draws
to calibrate each conventional cutoff and separate samples of $30{,}000$ null
and $30{,}000$ alternative draws for evaluation; the cycle test uses $2{,}000$
flips per draw.

\begin{table}[ht]
\centering
\caption{Empirical size and size-corrected power, concentrated design ($\lambda_n=0.286$).}
\label{tab:mc}
\scriptsize
\setlength{\tabcolsep}{3.2pt}
\begin{tabular}{lcccccccc}
\toprule
& \multicolumn{3}{c}{Size of nominal rule} &
  \multicolumn{2}{c}{Calibrated $5\%$ cutoff} &
  \multicolumn{3}{c}{Power at calibrated size} \\
\cmidrule(lr){2-4}\cmidrule(lr){5-6}\cmidrule(lr){7-9}
DGP & df-$t$ & HC2 & cycle & $|t_{\rm df}|$ & $|t_{\rm HC2}|$ & df-$t$ & HC2 & cycle \\
\midrule
H  & 0.054 & 0.093 & \textbf{0.049} & 2.009 & 2.383 & 0.934 & 0.848 & 0.576 \\
Ht & 0.315 & 0.205 & \textbf{0.049} & 4.651 & 3.158 & 0.275 & 0.511 & \textbf{0.520} \\
R  & 0.585 & 0.334 & \textbf{0.050} & 3.582 & 4.109 & \textbf{0.496} & 0.235 & 0.495 \\
\bottomrule
\end{tabular}
\par\smallskip
\begin{minipage}{0.96\textwidth}\footnotesize
Notes: H is homoskedastic Laplace; Ht adds a fourfold standard-deviation spike;
R replaces the error at that spike by a scaled Rademacher draw. The signal is
$\delta\sqrt{V_n}=3.5$. Conventional cutoffs are infeasible, DGP-specific
oracle calibrations, included to compare power at a common size rather than as a
proposed procedure. On the independent null evaluation sample, their rejection
rates after calibration range from $0.048$ to $0.051$. The exact cycle rule uses
its randomization $p$-value in both panels. Simulation code and the full-precision
output are included in the replication package.
\end{minipage}
\end{table}

The cycle test's size is within simulation error of $5\%$ in every configuration, including (R), where the df-$t$ rejects a true null $58.5\%$ of the time. Once the conventional tests are calibrated to the same size, the cycle test slightly leads HC2 under (Ht), essentially ties df-$t$ under (R), and pays for exactness under (H). Its power in the concentrated design is below the $\Phi(h\sqrt{\kap} - z)$ benchmark; when (P1) fails Remark~\ref{rem:p1} says that benchmark no longer applies, rather than imposing a universal power bound. In the \emph{diffuse} design, where (P1) holds ($\max_c b_c^2 / \sum_c b_c^2 = 0.093$, $\kap = 0.828$), at $\delta \sqrt{V_n} = 3.5$ the df-$t$ rejects at $0.935$ against a theoretical $0.938$, and the cycle test at $0.862$ against the Theorem~\ref{thm:power} benchmark $\Phi(3.5\sqrt{0.828} - 1.96) = 0.890$. The $2.8$ percentage-point cycle gap is a finite-$C$ deviation from the asymptotic benchmark, not Monte Carlo uncertainty. These simulations retain the preregistered 44-support construction used to generate their draws; the stronger deterministic v0.6.0 packing, reported separately in Table~\ref{tab:kappa}, raises capture to $0.834$ and is not substituted after seeing the simulation outcomes.

\subsection{A diffuse diagnostic check}\label{sec:diffusecheck}

The real F-score panel provides a useful negative control. Its design
diagnostics, $\lambda_n=0.0359$ and $N_{\mathrm{eff}}=83.1$, place it in the
diffuse regime; the log-return score diagnostic is also much smaller than in
Grunfeld ($\widehat\lambda_n^{\mathrm{score}}=0.133$). Conventional and exact
intervals both include zero. A diagnostic that always fires would not be useful;
this case shows the intended descriptive pass. Appendix~H
preserves the full workflow, raw-return comparison, and dependence analysis.

\section{Discussion}\label{sec:discussion}

\paragraph{Summary.} The paper's question was operational: when identifying variation is
concentrated, is there still a test valid at every sample size, and how much does it cost? The
answer is constructive rather than merely negative. A nuisance-annihilating contrast turns
sign-flipping into an exact symmetry regardless of leverage (Theorem~\ref{thm:exact}); in the
two-way case these contrasts are exactly the cycles of the mobility graph
(Proposition~\ref{prop:cyclespace}), so the same movement that identifies the treatment effect is
what makes exact inference possible. The cost is power, and it is observable before any outcome is
seen: the capture ratio $\kap$ (Theorem~\ref{thm:power}) and the packing algorithm of
Section~\ref{sec:packing} turn an otherwise abstract trade-off into a number a practitioner can
compute from the design alone. On the canonical corpus this distinction has teeth --- the exact
test matches or beats the calibrated conventional tests under heteroskedasticity while holding size
exactly (Section~\ref{sec:mc}), and structured packing recovers roughly twice the identifying
variation that naive packing would (Table~\ref{tab:kappa}).

\paragraph{Relation to existing methods.}\label{sec:related}

\begin{remark}[Infeasibility of the cyclic permutation test under saturation]\label{prop:cpt}
The cyclic permutation test of \citet{LeiBickel2021} for a linear hypothesis in the fixed-design model with $p$ regressors requires $n / p \ge 1/\alpha - 1$. In model \eqref{eq:model} with $p = d_n + 1$, this immediately requires $\rho_n := d_n / n \le \alpha/(1 - \alpha) - 1/n$; at $\alpha = 0.05$, $\rho_n \lesssim 0.053$. Hence the test is unavailable in the saturated regime $\rho_n$ bounded away from $\alpha/(1-\alpha)$, precisely where fixed-effect leverage concerns are most acute.
\end{remark}

Appendix~D gives full comparisons. Few-cluster sign-change methods
\citep{CRS2017,CSS2021,Toulis2022} estimate nuisance parameters and obtain
asymptotic invariance under homogeneity restrictions; exact annihilation instead
gives finite-sample validity under explicit symmetry. Existing fixed-design
randomization tests require exchangeability, $p<n/2$, or independence between
tested and nuisance regressors \citep{LeiBickel2021,WenWangWang2025,
LiZhouZhang2026,DHaultfoeuilleTuvaandorj2024}. These conditions require separate checks: our Grunfeld and F-score examples satisfy $p<n/2$, while sign symmetry allows heteroskedasticity that full exchangeability generally excludes.
Pouliot's exact $t$-test \citep{Pouliot2026} provides the general
orbit-projection construction underlying our matched-weight projection test;
Proposition~\ref{prop:exactt} gives the support-sign specialization and its local
rank formula. Proposition~\ref{prop:asymmetry} supplies a direct asymptotic
robustness argument for these sign flips. Network FE work develops leave-out inference under diffuse scores
\citep{JochmansWeidner2019,KSS2020,Jochmans2022}; our method covers the
concentrated regime. The concentration boundary itself has a counterpart in the
weak-identification literature, where inference must remain valid uniformly as
identification strength drifts toward a degenerate limit
\citep{AndrewsCheng2014,AndrewsGuggenberger2017}; the mechanism there is a
vanishing Jacobian rather than a vanishing effective sample size, but the
diagnosis is the same: validity that holds pointwise at every fixed parameter
value can still fail uniformly at the boundary. \citet{Crippa2025} tests
additivity with cycles; we use additivity to conduct exact inference on $\beta$.
The same preservation issue appears in dynamic hypergraphs. A companion
specification study 
\citep{HypergraphSpecification2026}\ shows that permuting overlapping-state
residuals destroys structural zeros implied by the Markov null and makes a
specification test reject a correct model with probability tending to one;
conditioning on the shared state and bootstrapping from the fitted null repairs
the reference distribution. Here the corresponding safeguard is exact
annihilation plus a sign-flip group that preserves the null symmetry.

\paragraph{Open problems.}

Three design questions remain open. First, splitting a dominant biconnected block into balanced supports calls for approximation guarantees. Second, the value of optimal variance-weighted packing on realistic graph sequences, and its connection to spectral connectivity, is unknown. Third, extending the construction to factor models beyond the conditional-symmetry case of Section~\ref{sec:dep} may require non-diagonal randomization groups. Proposition~\ref{prop:asymmetry} establishes asymptotic validity under asymmetric independent scores satisfying its negligibility condition, following the robustness principle of Pouliot's exact $t$-test. Robustness to asymmetry with a dominant score remains unresolved. Multiway designs and vector-valued $\beta$ retain the sign-invariance argument of Lemma~\ref{lem:inv}, although their construction and power require separate work.

\section*{Acknowledgments}

I thank Guillaume A.~Pouliot for pointing out that the matched-weight
projection-contrast statistic is a block-sign specialization of his exact
$t$-test, for supplying the projection and rank calculation, and for clarifying
the connection to asymptotic robustness without symmetry. He also drew my
attention to his joint work with Kirill Ponomarev on more general linear
actions.

\section*{Declarations}

\noindent\textbf{Competing interests.} The author declares no competing interests.

\smallskip\noindent\textbf{Funding.} This research did not receive any specific
grant from funding agencies in the public, commercial, or not-for-profit sectors.

\smallskip\noindent\textbf{Generative AI and AI-assisted technologies.}
During the preparation of this work the author used an AI proof-checking tool
based on Opus 5 (Anthropic) and GPT-5.5 (OpenAI) to perform supplementary
consistency checks on the mathematical arguments. The author verified all
content, reviewed and edited it as needed, and takes full responsibility for
the content of the published article.

\smallskip\noindent\textbf{Data availability.}
The Grunfeld application uses the public-domain canonical 11-firm file at a
pinned \texttt{statsmodels} revision; the matched employer--employee network is
the public testing extract of \citet{KSS2020}, downloaded by the replication
code; and the F-score panel analyzed in Appendix~H ships with that
code. The Grunfeld and F-score panels are also bundled in both software
packages. Code reproducing every
computation, including the cycle-packing implementation and the verification
suite, is distributed in \texttt{PanelAdequacy.jl}
(\url{https://github.com/profsms/PanelAdequacy.jl}) and the companion
\textsf{R} package \texttt{panelcert}
(\url{https://github.com/profsms/panelcert}).

\appendix

\section{Proofs}\label{app:proofs}

\subsection{Proof of Proposition~\ref{thm:conv} (Convolution limit; classical)}

\emph{Restatement.} Let Assumption~\ref{ass:err} hold, and suppose there exist an integer $J \ge 0$ and constants $a_1, \ldots, a_J$ such that $a_{nj} \to a_j$ for each $j \le J$ (after relabeling observations) while $\max_{i > J} |a_{ni}| \to 0$. Write $c := \sum_{j \le J} a_j^2 \in [0, 1]$. Then
\[
S_n := \sum_{i=1}^n a_{ni}\, \varepsilon_i \;\Longrightarrow\; \sum_{j=1}^{J} a_j\, \varepsilon_j^{*} \;+\; \sqrt{(1-c)}\;\sigma Z,
\]
where $\varepsilon_1^*, \ldots, \varepsilon_J^*$ are i.i.d.\ draws from $F$ and $Z \sim N(0,1)$ is independent of them.

\begin{proof}
Split $S_n = H_n + T_n$ with $H_n := \sum_{j \le J} a_{nj} \varepsilon_j$ and $T_n := \sum_{i > J} a_{ni} \varepsilon_i$; the two are independent for every $n$. Since the $a_{nj}$ are deterministic and converge, $H_n \to \sum_{j \le J} a_j \varepsilon_j$ almost surely, and the limit has the law of $\sum_j a_j \varepsilon_j^*$. For the tail, $\Var(T_n) = \sigma^2 (1 - \sum_{j \le J} a_{nj}^2) \to \sigma^2 (1 - c)$. If $c = 1$, $T_n \to 0$ in $L^2$ and the result follows by Slutsky. If $c < 1$, set $m_n := \max_{i > J} |a_{ni}| \to 0$ and verify Lindeberg's condition for the triangular array $\{a_{ni} \varepsilon_i\}_{i > J}$: for any $\delta > 0$,
\[
\sum_{i > J} \E\big[a_{ni}^2 \varepsilon_i^2\, \one\{|a_{ni} \varepsilon_i| > \delta\}\big]
\;\le\; \Big(\sum_{i > J} a_{ni}^2\Big)\, \E\big[\varepsilon^2\, \one\{|\varepsilon| > \delta / m_n\}\big] \;\to\; 0
\]
by dominated convergence, since $\E \varepsilon^2 < \infty$ and $\delta/m_n \to \infty$. Hence $T_n \Rightarrow N(0, \sigma^2(1-c))$ by Lindeberg--Feller. Independence of $H_n$ and $T_n$ gives joint convergence of the pair, and the continuous-mapping theorem applied to the sum yields the stated convolution.
\end{proof}

\subsection{Proof of Proposition~\ref{ex:rade} (No universally exact fixed critical value under concentration)}

\emph{Restatement.} No fixed cutoff attains limiting rejection probability $\alpha$ at every symmetric law of variance $\sigma^2$. The supremum over these laws of the pointwise limit superior at cutoff $\varkappa>0$ is $\min\{1,\varkappa^{-2}\}$.

\begin{proof}
Fix $\varkappa>0$ and $0<p<\min\{1,\varkappa^{-2}\}$. Let $F_p$ place mass $p/2$ at each of $\pm\sigma/\sqrt p$ and mass $1-p$ at zero. This law is symmetric with mean zero and variance $\sigma^2$. Since $1/\sqrt p>\varkappa$, neither boundary point $\pm\sigma\varkappa$ is an atom. Weak convergence therefore gives
\[
 \lim_n\Pp_{F_p}(|T_n|>\varkappa)
 =\Pp_{F_p}(|\varepsilon^*/\sigma|>\varkappa)=p.
\]
The interval of admissible $p$ has positive length, so for any $\alpha\in(0,1)$ it contains a $p\ne\alpha$. This proves the failure of universal exact calibration.

For an arbitrary $F$ in the class, the closed-set bound in the Portmanteau theorem and the second-moment inequality give
\[
 \limsup_n\Pp_F(|T_n|>\varkappa)
 \le \Pp_F(|\varepsilon^*/\sigma|\ge\varkappa)
 \le \min\{1,\varkappa^{-2}\}.
\]
The use of a weak inequality inside the limiting probability is essential at atoms. Letting $p\uparrow\min\{1,\varkappa^{-2}\}$ in the preceding family proves sharpness of the supremum. Substituting $\varkappa\ge\alpha^{-1/2}$ yields the stated conservative control.
\end{proof}

\subsection{Proof of Lemma~\ref{lem:inv} (Sign-flip invariance)}

\emph{Restatement.} Let Assumption~\ref{ass:sym} hold, let $q_1, \ldots, q_C$ be an annihilating contrast system, and define the contrast scores $U_c := v_c'(Y - x \beta_0)$. Under $H_0$, for every $s \in \{-1, +1\}^C$,
\[
(U_1, \ldots, U_C) \;\overset{d}{=}\; (s_1 U_1, \ldots, s_C U_C), \qquad \text{conditional on } (x, D).
\]

\begin{proof}
Under $H_0$, $Y - x\beta_0 = D\gamma + \varepsilon$, so $U_c = v_c' D \gamma + v_c' \varepsilon = v_c' \varepsilon$ by Definition~\ref{def:contrast}(ii): the fixed effects are annihilated identically, with no estimation. Fix $s$ and define $\tilde\varepsilon$ by flipping the sign of $\varepsilon_{B_k}$ for every block $B_k$ contained in $\bigcup_{c : s_c = -1} S_c$ and leaving all other blocks unchanged. By Assumption~\ref{ass:sym}(i)--(ii), $\tilde\varepsilon \overset{d}{=} \varepsilon$. By (iii) and support-disjointness, $v_c' \tilde\varepsilon = s_c\, v_c' \varepsilon$ for every $c$. Hence $(s_c U_c)_c = (v_c'\tilde\varepsilon)_c \overset{d}{=} (v_c'\varepsilon)_c = (U_c)_c$.
\end{proof}

\subsection{Proof of Theorem~\ref{thm:exact} (Finite-sample exactness)}

\emph{Restatement.} Let the conditions of Lemma~\ref{lem:inv} hold, let $T : \R^C \to \R$ be any measurable test statistic, and let $\mathcal G = \{-1, +1\}^C$. Define the randomization $p$-value
\[
p \;:=\; \frac{1}{|\mathcal G|} \sum_{s \in \mathcal G} \one\big\{ T(s_1 U_1, \ldots, s_C U_C) \ge T(U_1, \ldots, U_C) \big\} .
\]
Then $\Pp(p \le \alpha \mid x, D) \le \alpha$ under $H_0$, for every $n$, every $\alpha \in (0, 1)$, and every error distribution satisfying Assumption~\ref{ass:sym}. The same holds for the Monte Carlo version based on $B$ i.i.d.\ uniform draws $s^{(1)}, \ldots, s^{(B)}$ from $\mathcal G$, drawn independently of $U$ conditional on $(x,D)$, with $p = \big(1 + \#\{b : T(s^{(b)} \circ U) \ge T(U)\}\big) / (1 + B)$.

\begin{proof}
$\mathcal G$ is a group acting on $\R^C$ by coordinatewise sign changes, and by Lemma~\ref{lem:inv} the null distribution of $U := (U_1, \ldots, U_C)$ is $\mathcal G$-invariant. The claim is then the standard randomization guarantee \citep[Theorem 15.2.1]{LehmannRomano2005}: for $g$ uniform on $\mathcal G$ independent of $U$, the pair $(U, g \circ U)$ is exchangeable in the sense that $g \circ U \overset{d}{=} U$ and the orbit statistics $\{T(s \circ U)\}_{s \in \mathcal G}$ are, conditionally on the orbit of $U$, an exchangeable collection containing $T(U)$; the rank of $T(U)$ among them is therefore stochastically dominated by a uniform rank, which yields $\Pp(p \le \alpha) \le \alpha$. The Monte Carlo statement is the exactness of random-transformation tests with the identity adjoined \citep{HemerikGoeman2018}.
\end{proof}

\subsection{Proof of Proposition~\ref{prop:exactt} (Block-sign specialization)}

For each $c$, the group average $2^{-C}\sum_s s_cg_sD_A$ equals $D_c$
embedded on $A_c$ and is zero on the other supports. Hence
$\mathcal H_A$ contains every embedded local column space
$\operatorname{col}(D_c)$. Conversely, every $g_sD_A$ has its restriction
to $A_c$ in that local column space. Thus
$\mathcal H_A=\bigoplus_c\operatorname{col}(D_c)$, with embeddings understood.
Taking orthogonal complements gives $Q_A=\bigoplus_cQ_c$ and
\eqref{eq:exactt-rank}. Since $Q_c$ is an orthogonal projector,
$z_c'x_c=\|z_c\|^2$, so $b_c=\|z_c\|$ for the stated orientation and
$z_c=b_cv_c$ on $A_c$. For $r=Y-x\beta_0$,
\[
 x_A'Q_Ag_sr_A=\sum_cs_cz_c'r_{A_c}
 =\sum_cs_cb_cU_c.
\]
Zero local projections contribute zero. If the local annihilating space is
one-dimensional, any normalized generator $v_c$ satisfies
$Q_cx_c=v_c(v_c'x_c)$, proving the orientation-independent cycle version.

\subsection{Proof of Proposition~\ref{prop:asymmetry} (Asymptotic validity)}

All probabilities and expectations in this proof are conditional on the design.
Under the null, $U_c=W_{nc}$. Write
$S_n=\sum_c\xi_{nc}$ and $\widehat{\mathcal V}_n=\sum_c\xi_{nc}^2$.
Independence, centering, and $L_n\to0$ give the Lyapunov CLT
$S_n/\sqrt{\mathcal V_n}\Rightarrow N(0,1)$. Also,
\[
 \E\widehat{\mathcal V}_n=\mathcal V_n,\qquad
 \frac{\operatorname{Var}(\widehat{\mathcal V}_n)}{\mathcal V_n^2}
 \le L_n\to0,
\]
so $\widehat{\mathcal V}_n/\mathcal V_n\to_p1$ and
$S_n/\sqrt{\widehat{\mathcal V}_n}\Rightarrow N(0,1)$.
Moreover, by Markov's inequality,
\[
 \frac{\sum_c\xi_{nc}^4}{\mathcal V_n^2}\to_p0,
 \qquad
 \frac{\max_c\xi_{nc}^2}{\widehat{\mathcal V}_n}
 \le\frac{(\sum_c\xi_{nc}^4)^{1/2}}{\widehat{\mathcal V}_n}\to_p0.
\]
Conditional on the observed scores, $\sum_cs_c\xi_{nc}$ is a sum of
independent Rademacher-weighted constants, with variance
$\widehat{\mathcal V}_n$. The last display gives its conditional Lindeberg
condition. Its normalized conditional distribution therefore converges to
$N(0,1)$ in probability, uniformly in the argument by continuity of $\Phi$.
The absolute-value randomization distribution converges to that of $|N(0,1)|$;
its atoms vanish and its $(1-\alpha)$-quantile tends to $z_{1-\alpha/2}$.
Combining this with the sampling CLT proves the rejection-probability claim,
including the nonrandomized tie convention. The event
$\widehat{\mathcal V}_n=0$ has probability tending to zero and is assigned
$p=1$. With $B_n\to\infty$, conditional empirical-cdf convergence gives the
same limit for sampled flips, and the plus-one adjustment vanishes.
Finally, $\widehat{\mathcal V}_n$ is constant over sign patterns, so
normalizing does not change any comparison in the unstudentized $p$-value.
The exactness claim under symmetry is Theorem~\ref{thm:exact}.

For the stated sufficient condition, write $\omega_{nc}^2=\E W_{nc}^2$.
If $\E W_{nc}^4\le K\omega_{nc}^4$ uniformly, then
\[
 L_n\le K\frac{\sum_c(b_c^2\omega_{nc}^2)^2}{\mathcal V_n^2}
 \le K\frac{\max_cb_c^2\omega_{nc}^2}{\mathcal V_n}\to0.
\]

\subsection{Proof of Lemma~\ref{lem:compat} (Compatible supports in two-way designs)}

\emph{Restatement.} Let the observations be the edges of the bipartite multigraph $G$ of
Section~\ref{sec:cycles}, and let a support system $A_1, \ldots, A_C$ be given.
\begin{enumerate}[label=(\alph*)]
\item \emph{Observation level} ($\mathcal P$ = singletons): every support system
is compatible.
\item \emph{Match level} ($\mathcal P$ = parallel-edge classes, i.e.\ arbitrary
dependence within a worker--firm match): $A_j$ is compatible iff it contains all
or none of the edges of each match. Every cycle through matches of multiplicity
one is compatible; a digon extracted from a match of multiplicity $k \ge 3$ is
not, but the whole match taken as one support is, with local cycle space of
dimension $k - 1$.
\item \emph{$\mathcal A$-vertex level} ($\mathcal P$ = workers, i.e.\ arbitrary
dependence across a worker's spells): $A_j$ is compatible iff it uses each
worker's edges all or none. In particular \emph{every four-cycle built from
two-period movers is automatically compatible}, because such a cycle uses both
edges of each of its two workers and those workers have no others; so is every
$2L$-cycle produced by the contraction of Proposition~\ref{prop:contract}.
\item \emph{$\mathcal B$-vertex level} ($\mathcal P$ = firms): $A_j$ is
compatible iff it uses each firm's edges all or none. Since firms are hubs, no
single cycle is compatible unless every firm it touches has no other edge;
compatibility forces supports at the level of groups of whole firms.
\item \emph{Period level} ($\mathcal P$ = calendar periods, when one vertex class
is time or country--time): symmetrically to (d), $A_j$ must use each period's
edges all or none.
\end{enumerate}

\begin{proof}
Condition (iii) is verbatim the statement that each $A_j$ is
$\mathcal P$-measurable, which gives (a) and the first sentence of each of
(b)--(e). For the second sentence of (b), the local cycle space of a match with
$k$ parallel edges is $\{q \in \R^k : \sum_e q_e = 0\}$, of dimension $k-1$. For
(c), a two-period mover contributes exactly two edges, both used by any
four-cycle containing her (Remark~\ref{prop:auto}), and the cycles of
Proposition~\ref{prop:contract} traverse each mover via both of her edges by
construction. For (d), a cycle visiting firm $f$ uses exactly two of $f$'s
edges, so it is $\mathcal P$-measurable only if $f$ has degree two.
\end{proof}

\subsection{Proof of Proposition~\ref{prop:pairsym} (Symmetry from antithetic pairing)}

\emph{Restatement.} Let $q$ be a contrast whose support admits a partition into pairs
$\{i_1, j_1\}, \ldots, \{i_m, j_m\}$ with $q_{i_k} = -q_{j_k}$ for every $k$.
Suppose the pair vectors $(\varepsilon_{i_k}, \varepsilon_{j_k})$, $k \le m$,
are mutually independent and each is exchangeable,
$(\varepsilon_{i_k}, \varepsilon_{j_k}) \overset{d}{=} (\varepsilon_{j_k}, \varepsilon_{i_k})$.
Then $q'\varepsilon$ is symmetric about zero. If an annihilating contrast system
consists of such contrasts with disjoint supports and the pairs are mutually
independent across contrasts as well, the conclusion of Lemma~\ref{lem:inv}---and
hence Theorem~\ref{thm:exact}---holds with Assumption~\ref{ass:sym}(ii) replaced
by this pairwise exchangeability. No symmetry, and indeed no moment, is required
of the marginal law of $\varepsilon$.

\begin{proof}
Write $q'\varepsilon = \sum_{k \le m} q_{i_k}(\varepsilon_{i_k} - \varepsilon_{j_k})$.
Exchangeability of the $k$th pair gives
$\varepsilon_{i_k} - \varepsilon_{j_k} \overset{d}{=} \varepsilon_{j_k} - \varepsilon_{i_k}
= -(\varepsilon_{i_k} - \varepsilon_{j_k})$, so each summand is symmetric about
zero; the summands are independent, and a sum of independent symmetric variables
is symmetric. For the system statement, disjointness of supports and
independence across pairs make the $C$ scores independent and each symmetric,
so their joint law is invariant under independent sign flips, which is the
conclusion of Lemma~\ref{lem:inv}; Theorem~\ref{thm:exact} used nothing else.
\end{proof}

\subsection{Proof of Theorem~\ref{thm:power} (Local power and Pitman efficiency)}

\emph{Restatement.} Let blocks be singletons: conditional on $(x,D)$ the errors are
independent, symmetric, $\E[\varepsilon_i^2] = \sigma^2$ for all $i$, and
$\sup_i \E[\varepsilon_i^4] \le \bar\kappa < \infty$. Let
$(q_c)_{c \le C_n}$ be annihilating contrast systems with
$C_n \to \infty$, write $\Sigma_n := \sum_{c \le C_n} b_c^2$, and suppose
$\Sigma_n\to\infty$. Consider
local alternatives
$\beta_n = \beta_0 + h \sigma / \Sigma_n^{1/2}$, $h \in \R$ fixed, tested
at level $\alpha$ by the full-enumeration randomization test of
Theorem~\ref{thm:exact}, or by its Monte Carlo version with the number of
sampled flips $B_n\to\infty$, using statistic $T(u)=|\sum_c b_cu_c|$.
\begin{enumerate}[label=(\roman*)]
\item If
\begin{equation*}\tag{P1}
\max_{c \le C_n} b_c^2 \Big/ \Sigma_n \;\longrightarrow\; 0 ,
\end{equation*}
then the power of the test converges to
$\Phi\big(h - z_{1 - \alpha/2}\big) + \Phi\big(-h - z_{1 - \alpha/2}\big)$.
\item If in addition
\begin{equation*}\tag{P2}
\lambda_n \;=\; \max_{i \le n} \tx_i^2 / V_n \;\longrightarrow\; 0 ,
\end{equation*}
then the infeasible oracle test that rejects when
$|\tx'(Y - x\beta_0)| / (\sigma \sqrt{V_n}) > z_{1-\alpha/2}$ has
asymptotic level $\alpha$ and attains the same limiting power along
$\beta_n = \beta_0 + h\sigma/\sqrt{V_n}$. Consequently, if
$\kap := \Sigma_n / V_n$ converges, the Pitman asymptotic relative
efficiency of the contrast test with respect to that oracle equals
$\lim \kap$.
\end{enumerate}

\begin{proof}
Write $\delta_n := \beta_n - \beta_0$, $W_c := v_c'\varepsilon$,
$U_c = b_c\delta_n + W_c$, and $S_n := \sum_c b_c U_c$. The $W_c$ are
independent across $c$ (disjoint supports, independent errors), mean zero,
with $\E[W_c^2] = \sigma^2\|v_c\|^2 = \sigma^2$. Expanding
$W_c = \sum_i v_{ci}\varepsilon_i$ and discarding the terms containing an
odd power of a single $\varepsilon_i$,
\begin{equation}\label{eq:fourth}
\E[W_c^4]
= \sum_i v_{ci}^4 \E[\varepsilon_i^4]
+ 3 \sum_{i \ne j} v_{ci}^2 v_{cj}^2 \sigma^4
\;\le\; \bar\kappa \sum_i v_{ci}^4 + 3\sigma^4 \Big(\sum_i v_{ci}^2\Big)^2
\;\le\; \bar\kappa + 3\sigma^4 \;=:\; \bar\kappa_U ,
\end{equation}
using $\sum_i v_{ci}^4 \le (\sum_i v_{ci}^2)^2 = \|v_c\|^4 = 1$. No
condition beyond $\|v_c\| = 1$ is needed for \eqref{eq:fourth}; in
particular the contrast supports may be arbitrarily large.

\emph{Step 1: sampling distribution of $S_n$.} Since
$\sum_c b_c^4 \le (\max_c b_c^2)\,\Sigma_n$, the Lyapunov ratio for
$\sum_c b_c W_c / (\sigma \Sigma_n^{1/2})$ obeys
\[
\frac{\sum_c b_c^4\, \E[W_c^4]}{\sigma^4 \Sigma_n^{2}}
\;\le\; \frac{\bar\kappa_U}{\sigma^4}\cdot\frac{\max_c b_c^2}{\Sigma_n}
\;\longrightarrow\; 0
\]
by \eqref{eq:P1} and \eqref{eq:fourth}. Hence
$S_n / (\sigma \Sigma_n^{1/2}) \Rightarrow N(h, 1)$, using
$\delta_n \Sigma_n / (\sigma \Sigma_n^{1/2}) = h$.

\emph{Step 2: randomization critical value.} For full enumeration, conditionally
on $U = (U_c)_c$ the randomization statistic is
$S_n^*(s) = \sum_c s_c w_c$ with $w_c := b_c U_c$ and i.i.d.\ Rademacher
$s_c$. Abbreviate $\theta_n := \max_c b_c^2/\Sigma_n$, so
$\theta_n \to 0$, and note $\delta_n^2 = h^2\sigma^2/\Sigma_n$.

\emph{(2a) Conditional scale.} We claim
$\sum_c w_c^2 / (\sigma^2 \Sigma_n) \to_p 1$. For the mean,
$\E[U_c^2] = \sigma^2 + b_c^2\delta_n^2$, so
\[
\E\Big[\sum_c w_c^2\Big]
= \sigma^2 \Sigma_n + \delta_n^2 \sum_c b_c^4
\le \sigma^2\Sigma_n + \delta_n^2 (\max_c b_c^2)\Sigma_n
= \sigma^2\Sigma_n\big(1 + h^2\theta_n\big),
\]
which is $\sigma^2\Sigma_n(1+o(1))$. For the variance, the $U_c$ are
independent across $c$, and
$\E[U_c^4] \le 8(b_c^4\delta_n^4 + \E[W_c^4]) \le 8(b_c^4\delta_n^4 + \bar\kappa_U)$,
so
\[
\frac{\Var\big(\sum_c w_c^2\big)}{\sigma^4\Sigma_n^2}
\;\le\; \frac{\sum_c b_c^4\, \E[U_c^4]}{\sigma^4 \Sigma_n^2}
\;\le\; \frac{8\bar\kappa_U}{\sigma^4}\,\theta_n
\;+\; \frac{8\delta_n^4 \sum_c b_c^8}{\sigma^4\Sigma_n^2}
\;\le\; \frac{8\bar\kappa_U}{\sigma^4}\,\theta_n + 8h^4 \theta_n^3
\;\longrightarrow\; 0 ,
\]
where the last bound uses
$\sum_c b_c^8 \le (\max_c b_c^2)^3\Sigma_n$ and
$\delta_n^4 = h^4\sigma^4/\Sigma_n^2$. Chebyshev gives the claim.

\emph{(2b) Conditional CLT.} The same two bounds give
\[
\E\big[\sum_c w_c^4\big] \le 8\bar\kappa_U (\max_c b_c^2)\Sigma_n
+ 8\delta_n^4 (\max_c b_c^2)^3\Sigma_n
= \Sigma_n^2\, O(\theta_n),
\]
so $\sum_c w_c^4 = o_p(\Sigma_n^2)$ by Markov, while
$(\sum_c w_c^2)^2 = \sigma^4\Sigma_n^2(1+o_p(1))$ by (2a). Since
$\max_c w_c^2 \le (\sum_c w_c^4)^{1/2}$,
\[
\frac{\max_c w_c^2}{\sum_c w_c^2}
\;\le\; \Big(\frac{\sum_c w_c^4}{(\sum_c w_c^2)^2}\Big)^{1/2}
\;\longrightarrow_p\; 0 ,
\]
which is the Lindeberg condition for a weighted Rademacher sum. Hence the
conditional law of $S_n^*/(\sum_c w_c^2)^{1/2}$ converges weakly to
$N(0,1)$ in probability; because $\Phi$ is continuous, P\'olya's theorem
upgrades this to uniform convergence of the conditional distribution
functions, and therefore to convergence of conditional quantiles. Hence
$|S_n^*|/(\sum_c w_c^2)^{1/2}$ converges conditionally to $|N(0,1)|$,
whose $(1-\alpha)$-quantile is $z_{1-\alpha/2}$; with (2a), the
conditional $(1-\alpha)$-quantile of $|S_n^*|$ equals
$\sigma \Sigma_n^{1/2}\,(z_{1-\alpha/2} + o_p(1))$.
For the Monte Carlo version, the same conclusion follows when $B_n\to\infty$
from conditional Glivenko--Cantelli convergence of the empirical sign-flip
distribution. With fixed $B$, Theorem~\ref{thm:exact}'s plus-one validity remains
finite-sample exact, but the displayed Gaussian local-power formula need not be
its limiting power.

\emph{Step 3: parts (i) and (ii).} The test rejects when
$|S_n|$ exceeds that quantile; dividing by $\sigma\Sigma_n^{1/2}$ and
combining Step~1 with Step~2 by Slutsky, the rejection probability
converges to
$\Pp(|N(h,1)| > z_{1-\alpha/2})
= \Phi(h - z_{1-\alpha/2}) + \Phi(-h - z_{1-\alpha/2})$, which is~(i).
For~(ii), the oracle statistic is
$\tx'\varepsilon/(\sigma\sqrt{V_n}) = \sum_i (\tx_i/\sqrt{V_n})\,\varepsilon_i/\sigma$,
a weighted sum with $\max_i \tx_i^2/V_n = \lambda_n$; its Lyapunov ratio is
$\sum_i (\tx_i^2/V_n)^2 \E[\varepsilon_i^4]/\sigma^4 \le (\bar\kappa/\sigma^4)\lambda_n \to 0$
under \eqref{eq:P2}, so it is asymptotically $N(0,1)$ under $H_0$, the
Gaussian critical value is asymptotically correct, and along
$\beta_0 + h\sigma/\sqrt{V_n}$ its power has the same limit. The contrast
test detects drifts of size $h\sigma/\Sigma_n^{1/2}$ and the oracle drifts
of size $h\sigma/V_n^{1/2}$ at equal limiting power. The ratio of
efficacies---$\Sigma_n/\sigma^2$ for the contrast test against
$V_n/\sigma^2$ for the oracle---is therefore $\Sigma_n/V_n = \kap$, which
is the Pitman ARE; equivalently, when $\kap>0$ the contrast test needs
$1/\kap$ times the squared drift to match the oracle's power, and when
$\kap=0$ (captured information asymptotically negligible relative to $V_n$)
no finite multiple suffices.
\end{proof}

\subsection{Proof of Proposition~\ref{prop:het-supp} (Optimal weighting and heteroskedastic efficiency)}

\emph{Restatement.} Let blocks be singletons, the errors independent, symmetric, with
covariance $\Omega$ as above and $\sup_i \E[\varepsilon_i^4] \le \bar\kappa$.
For a weight vector $a = (a_c)_{c \le C_n}$ not identically zero (individual
$a_c$ may be zero) with $\sum_c a_cb_c\ne0$ --- error-independent, but possibly
depending on $\Omega$ as well as the design, hence infeasible unless $\Omega$
is known, as for the optimal weights of part~(a) --- let $T_a(u) = |\sum_c a_c u_c|$,
which is exact by Theorem~\ref{thm:exact} for every choice of $a$, and let
\[
e_n(a) \;:=\; \Big(\sum_c a_c b_c\Big)^2 \Big/ \sum_c a_c^2 \omega_c^2
\]
be its efficacy. Then:
\begin{enumerate}[label=(\alph*)]
\item $e_n(a) \le \mathcal{E}_n$ for every $a$, with equality if and only if
$a_c \propto b_c/\omega_c^2$.
\item Suppose $e_n(a)\to\infty$ and both
\begin{equation*}\tag{P1a}
\max_{c \le C_n} a_c^2\omega_c^2 \Big/ \sum_c a_c^2 \omega_c^2
\;\longrightarrow\; 0
\end{equation*}
and
\begin{equation*}\tag{P1b}
\sum_c a_c^2 b_c^2 \Big/ \Big(\sum_c a_c b_c\Big)^{2}
\;\longrightarrow\; 0 .
\end{equation*}
Then along $\beta_n = \beta_0 + h\, e_n(a)^{-1/2}$ the power of the
level-$\alpha$ test based on $T_a$ converges to
$\Phi(h - z_{1-\alpha/2}) + \Phi(-h - z_{1-\alpha/2})$.
Condition \eqref{eq:P1b} is not implied by \eqref{eq:P1a} and cannot be
dropped: see Remark~\ref{rem:p1b}. For the optimal weights of part~(a),
\eqref{eq:P1a} implies \eqref{eq:P1b}; and when $\Omega = \sigma^2 I$ and
$a_c = b_c$, both reduce to \eqref{eq:P1}, so Theorem~\ref{thm:power}(i)
is the special case it appears to be.
\item $\mathcal{E}_n \le \mathcal{V}_\Omega$, so
\[
\kap^{\mathrm{het}} \;:=\; \mathcal{E}_n / \mathcal{V}_\Omega \;\in\; [0,1] ,
\]
with equality if and only if $M_{\Omega^{-1/2}D}\Omega^{-1/2}x$ lies in
$\mathrm{span}\{\Omega^{1/2}v_c/\omega_c\}_c$. If moreover $\mathcal E_n\to\infty$,
$\max_i (M_{\Omega^{-1/2}D}\Omega^{-1/2}x)_i^2 / \mathcal{V}_\Omega \to 0$ and
$\kap^{\mathrm{het}}$ converges, then $\lim \kap^{\mathrm{het}}$ is the Pitman
asymptotic relative efficiency of the optimally weighted contrast test
with respect to the GLS oracle. Under homoskedasticity
$\kap^{\mathrm{het}} = \kap$.
\end{enumerate}

\begin{proof}
(a) Write $\sum_c a_c b_c = \sum_c (a_c \omega_c)(b_c/\omega_c)$ and apply
Cauchy--Schwarz:
$(\sum_c a_c b_c)^2 \le (\sum_c a_c^2\omega_c^2)(\sum_c b_c^2/\omega_c^2)$,
with equality iff $a_c\omega_c \propto b_c/\omega_c$, i.e.\
$a_c \propto b_c/\omega_c^2$.

(b) The proof of Theorem~\ref{thm:power}(i) goes through with $b_c$
replaced by $a_c$ in the weights, $\sigma^2$ by $\omega_c^2$, and
$\Sigma_n$ by $A_n := \sum_c a_c^2\omega_c^2$, but the roles of weight and
loading, which coincide there, now separate, and the two hypotheses enter
at different points. The fourth-moment bound \eqref{eq:fourth} becomes
$\E[W_c^4] \le \bar\kappa + 3\sigma_{\max}^4 =: \bar\kappa_U'$, again using
$\|v_c\| = 1$.

Write $\delta_n = h\,e_n(a)^{-1/2}$, so
$\delta_n^2 = h^2 A_n / (\sum_c a_c b_c)^2$, and $w_c := a_c U_c$.
\emph{Step 1} is unchanged: since $\omega_c^2 \ge \sigma_{\min}^2 > 0$,
\[
\frac{\sum_c a_c^4\,\E[W_c^4]}{A_n^2}
\;\le\; \frac{\bar\kappa_U'}{\sigma_{\min}^4}\,
\frac{\max_c a_c^2\omega_c^2}{A_n} \;\longrightarrow\; 0
\]
by \eqref{eq:P1a}, using $a_c^4 \le (a_c^2\omega_c^2)^2/\sigma_{\min}^4$;
hence $\sum_c a_c U_c / A_n^{1/2} \Rightarrow N(\pm h, 1)$, the sign being
that of $\sum_c a_c b_c$ and immaterial for the two-sided test.

\emph{Step 2a} is where \eqref{eq:P1b} is needed. Now
$\E[U_c^2] = \omega_c^2 + b_c^2\delta_n^2$, so
\[
\frac{\E\big[\sum_c w_c^2\big]}{A_n}
\;=\; 1 + \frac{\delta_n^2 \sum_c a_c^2 b_c^2}{A_n}
\;=\; 1 + h^2\,\frac{\sum_c a_c^2 b_c^2}{\big(\sum_c a_c b_c\big)^2} ,
\]
which tends to $1$ precisely under \eqref{eq:P1b}---and, absent it, the
conditional scale is inflated by a non-vanishing factor while the observed
statistic is not, which is exactly how the test loses power. For the
variance, $\E[U_c^4] \le 8(b_c^4\delta_n^4 + \bar\kappa_U')$ gives
\[
\frac{\Var\big(\sum_c w_c^2\big)}{A_n^2}
\;\le\; \frac{\sum_c a_c^4 \E[U_c^4]}{A_n^2}
\;\le\; \frac{8\bar\kappa_U'}{\sigma_{\min}^4}\,
\frac{\max_c a_c^2\omega_c^2}{A_n}
\;+\; 8h^4 \left(\frac{\sum_c a_c^2b_c^2}{(\sum_c a_cb_c)^2}\right)^{\!2}
\;\longrightarrow\; 0 ,
\]
using $\sum_c a_c^4 b_c^4 \le (\sum_c a_c^2 b_c^2)^2$ for the second term.
Chebyshev then gives $\sum_c w_c^2 / A_n \to_p 1$. \emph{Step 2b} follows
from the same two bounds, since
$\E[\sum_c w_c^4] = \sum_c a_c^4\E[U_c^4] = A_n^2\,o(1)$ by the display
above, and Step~3 is unchanged.

For the optimal weights $a_c \propto b_c/\omega_c^2$ of part~(a), put
$t_c := b_c^2/\omega_c^2$, so that $a_c^2\omega_c^2 = t_c$ and
$\sum_c a_c b_c = \sum_c t_c = \mathcal{E}_n$; then \eqref{eq:P1a} reads
$\max_c t_c/\mathcal{E}_n \to 0$, and
$\sum_c a_c^2b_c^2/(\sum_c a_cb_c)^2 = \sum_c t_c^2/\mathcal{E}_n^2
\le \max_c t_c/\mathcal{E}_n$, so \eqref{eq:P1b} follows. When
$\Omega = \sigma^2I$ and $a_c = b_c$, \eqref{eq:P1a} is \eqref{eq:P1} and
$\sum_c b_c^4/\Sigma_n^2 \le \max_c b_c^2/\Sigma_n$ gives \eqref{eq:P1b}.

(c) Put $z := \Omega^{-1/2}x$ and $g_c := \Omega^{1/2}v_c/\omega_c$. Then
$\|g_c\| = 1$, and for $c \ne c'$ we have
$g_c'g_{c'} = v_c'\Omega v_{c'}/(\omega_c\omega_{c'}) = 0$ because $\Omega$
is diagonal and the supports $S_c$ are disjoint; so $\{g_c\}$ is
orthonormal. Moreover
$g_c'\Omega^{-1/2}D = \omega_c^{-1} v_c'D = 0$, so
$g_c \in \mathrm{col}(\Omega^{-1/2}D)^\perp$. Since
$g_c'z = v_c'x/\omega_c = b_c/\omega_c$, Bessel's inequality applied to the
orthonormal system $\{g_c\}$ inside that subspace gives
\[
\mathcal{E}_n = \sum_c (g_c'z)^2 = \sum_c \big(g_c'M_{\Omega^{-1/2}D}z\big)^2
\;\le\; \big\|M_{\Omega^{-1/2}D}z\big\|^2 = \mathcal{V}_\Omega ,
\]
with the stated equality condition. The efficacy of the GLS oracle is
$\mathcal{V}_\Omega$ and its Lindeberg condition is the displayed leverage
condition on $M_{\Omega^{-1/2}D}\Omega^{-1/2}x$; the ARE claim then
follows exactly as in Theorem~\ref{thm:power}(ii), comparing efficacies
$\mathcal{E}_n$ and $\mathcal{V}_\Omega$. Setting $\Omega = \sigma^2I$ gives
$\omega_c^2 = \sigma^2$, $\mathcal{E}_n = \Sigma_n/\sigma^2$,
$\mathcal{V}_\Omega = V_n/\sigma^2$.
\end{proof}

\subsection{Proof of Proposition~\ref{prop:cyclespace} (Annihilating contrasts are the cycle space)}

\emph{Restatement.} Let $\mathcal Z := \{q \in \R^n : q'D = 0\}$.
\begin{enumerate}[label=(\alph*)]
\item $\mathcal Z$ equals the circulation (cycle) space of $G$ under any orientation of its edges from $\mathcal A$ to $\mathcal B$; in particular $\dim \mathcal Z = n - |\mathcal A| - |\mathcal B| + \#\{\text{components of } G\}$.
\item For any closed walk $e^{(1)}, e^{(2)}, \ldots, e^{(2L)}$ in $G$ that traverses distinct edges, the alternating vector $q$ with $q_{e^{(k)}} = (-1)^{k+1}$ and zeros elsewhere lies in $\mathcal Z$; for a pair of parallel edges $e, e'$, the digon vector $q_e = 1, q_{e'} = -1$ lies in $\mathcal Z$; and vectors of these two types span $\mathcal Z$.
\item $V_n = x' M_D x = \| \Pi_{\mathcal Z}\, x \|^2$, where $\Pi_{\mathcal Z}$ is the orthogonal projection onto $\mathcal Z$.
\end{enumerate}

\begin{proof}
(a) Orient every edge from its $\mathcal A$-endpoint to its $\mathcal B$-endpoint and let $M$ be the signed incidence matrix, $M_{e,v} = +1$ if $v = a(e)$, $-1$ if $v = b(e)$, $0$ otherwise. The condition $q'D = 0$ says $\sum_{e \ni v} q_e = 0$ at every vertex $v$. Because $G$ is bipartite and the orientation is uniform $\mathcal A \to \mathcal B$, each column of $M$ is either equal to the corresponding column of $D$ (for $v \in \mathcal A$) or its negative (for $v \in \mathcal B$); a constant sign per column does not change the kernel, so $\ker(D') = \ker(M')$, which is by definition the circulation space of the oriented multigraph. Its dimension is the cycle rank $n - |V| + \#\text{comp}$.

(b) Alternating $\pm 1$ along a closed walk assigns, at each visited vertex, values $+1$ and $-1$ to the two incident walk edges (walks in a bipartite graph have even length, and consecutive edges share a vertex with opposite alternation signs), so the unsigned vertex sums vanish; digons are the two-edge case. That such vectors span is the standard fact that fundamental cycles of any spanning forest form a basis of the circulation space \citep[e.g.,][\S 1.9]{Diestel2017}, and every fundamental cycle in a bipartite multigraph is an alternating closed walk or a digon.

(c) $M_D$ is the orthogonal projection onto $\mathrm{col}(D)^\perp = \ker(D') = \mathcal Z$.
\end{proof}

\subsection{Proof of Proposition~\ref{prop:contract} (Contraction principle)}

\emph{Restatement.} Remove stayers' edges via digons. Contract each remaining two-period mover $w$ with firms $\{f, f'\}$ into a single edge $\{f, f'\}$ of a \emph{firm multigraph} $G_F$, carrying the value $w$'s worker contrast. Then: (a) edge-disjoint cycle families of $G_F$ correspond bijectively to edge-disjoint bipartite cycle families through those movers, a cycle through $L$ firms corresponding to a bipartite $2L$-cycle with contrast value $\big(\sum_{k} \pm\, w_k\big)/\sqrt{2L}$ for the appropriate alternating signs; (b) parallel edges of $G_F$ are precisely firm-pair mover pairs, whose digons in $G_F$ are the four-cycles of Remark~\ref{prop:auto}. Consequently, the digon-first algorithm applies recursively: digons in $G$, then digons in $G_F$ (four-cycles in $G$), then longer cycles of $G_F$.

\begin{proof}
(a) A cycle $f_1 - f_2 - \cdots - f_L - f_1$ in $G_F$ traversing movers $w_1, \ldots, w_L$ lifts to the closed walk $f_1, w_1, f_2, w_2, \ldots, f_L, w_L, f_1$ in the bipartite graph, alternating firms and workers, of edge length $2L$ with all edges distinct because each mover contributes her own two edges and movers are distinct. The alternating $\pm 1$ contrast on this bipartite cycle assigns opposite signs to each mover's two edges, so its inner product with $\tx$ telescopes into $\pm$ sums of worker contrasts $w_k$; normalization is by $\sqrt{2L}$. Conversely, a bipartite cycle through two-period movers visits each such worker via both her edges (a two-period mover has no other edges), so it projects to a cycle of $G_F$. Edge-disjointness is preserved in both directions since bipartite edges partition by mover. (b) Immediate from the definitions, matching the contrast values: a $G_F$-digon on movers $i, j$ has value $(w_i - w_j)/\sqrt{2 \cdot 2}$ after lifting, which is the four-cycle value $(w_i - w_j)/2$.
\end{proof}

\subsection{Proof of Lemma~\ref{lem:maxgroup} (Maximality of the block-flip group among diagonal sign transformations)}

\emph{Restatement.} For $\eta \in \{\pm 1\}^E$, the map $\mathrm{diag}(\eta)$ satisfies
$\mathrm{diag}(\eta)\,\mathcal Z \subseteq \mathcal Z$ if and only if
$\eta$ is constant on the edge set of every nontrivial block of $G$
(its values on bridges being unrestricted). Consequently the group
\[
\mathcal G_{\max} := \{\eta : \mathrm{diag}(\eta)\mathcal Z \subseteq \mathcal Z\}
\big/ \{\eta \equiv 1 \text{ on all cycles}\}
\;\cong\; \{\pm 1\}^{B},
\]
one flip per nontrivial block. It is therefore the largest group of
\emph{diagonal, edgewise} $\pm 1$ transformations that leaves every
contrast in $\mathcal Z$ nuisance-free;
Proposition~\ref{prop:orbitmax} gives the corresponding statement for flip
groups required only to respect the particular contrast system in use.
The scope of the claim should be read literally: it concerns the class
$\{\mathrm{diag}(\eta) : \eta \in \{\pm 1\}^E\}$ and nothing else. It says
nothing about randomization groups acting non-diagonally (rotations,
permutations, block-orthogonal maps), about conditional procedures, or
about exact tests not built from a group at all; see
Remark~\ref{rem:openimposs} and Remark~\ref{rem:scopeimposs}.

\begin{proof}
($\Leftarrow$) If $\eta \equiv s_b$ on $E_b$, then for $z = \sum_b z_b$
with $z_b \in \mathcal Z_b$, $\mathrm{diag}(\eta) z = \sum_b s_b z_b \in
\mathcal Z$ by \eqref{eq:blockdecomp}.

($\Rightarrow$) Let $e \ne e'$ lie in the same nontrivial block
$\mathcal B$. We first record that $\mathcal B$ contains a simple cycle
$C$ through both. \citet[Prop.~3.1.1]{Diestel2017} gives this for
$2$-connected \emph{simple} graphs, and the multigraph form follows by a
two-case reduction. If $e$ and $e'$ are parallel, sharing endpoints $u,v$,
then $C = \{e, e'\}$ is itself a digon and there is nothing to prove.
Otherwise $e$ and $e'$ have distinct endpoint pairs, so their images
$\bar e \ne \bar e'$ in the underlying simple graph $\bar{\mathcal B}$ are
distinct edges. Since suppressing parallel copies changes neither the
vertex set nor the cut vertices of $\mathcal B$, and $\mathcal B$ is a
nontrivial block with $|V(\mathcal B)| \ge 3$ in this case,
$\bar{\mathcal B}$ is $2$-connected; the cited proposition supplies a
simple cycle $\bar C \subseteq \bar{\mathcal B}$ through $\bar e$ and
$\bar e'$. Lifting $\bar C$ to $\mathcal B$ by choosing the copy $e$ for
$\bar e$, the copy $e'$ for $\bar e'$, and an arbitrary copy for each
remaining edge yields a simple cycle $C$ of $\mathcal B$ containing $e$
and $e'$. (The two cases are exhaustive: a nontrivial block on two
vertices consists precisely of parallel edges.)

Let $z_C \in \mathcal Z$ be the alternating vector of $C$, supported
exactly on $E(C)$. Since $\mathrm{diag}(\eta) z_C \in \mathcal Z$ and its
support lies in $E(C)$, it belongs to $\mathcal Z \cap \R^{E(C)}$, the
cycle space of the sub-multigraph consisting of the edges of $C$ alone.
For a simple cycle---the digon included---that space has dimension
$|E(C)| - |V(C)| + 1 = 1$, i.e., equals $\mathrm{span}(z_C)$. Hence
$\mathrm{diag}(\eta) z_C = c\, z_C$ with $c \in \{\pm 1\}$, and since
$(z_C)_e \ne 0$ on all of $E(C)$, $\eta_e = c$ for every $e \in E(C)$;
in particular $\eta_e = \eta_{e'}$. As $e, e'$ were arbitrary in the
block, $\eta$ is constant on it.

The final claim is immediate: any $\eta$ with $\mathrm{diag}(\eta)
\mathcal Z \subseteq \mathcal Z$ lies in $\mathcal G_{\max}$ by
definition, and the two implications just proved identify that set with
the block-constant patterns.
\end{proof}

\subsection{Proof of Proposition~\ref{prop:orbitmax} (Flip groups adapted to a contrast system)}

\emph{Restatement.} Let $\mathcal V \subset \mathcal Z$ be a finite contrast system and let
$\mathcal G \subseteq \{\pm 1\}^E$ be a group of sign patterns
\emph{admissible} for $\mathcal V$, meaning $\mathrm{diag}(\eta) v \in
\mathcal Z$ for all $\eta \in \mathcal G$ and $v \in \mathcal V$, so that
every flipped statistic $v'\mathrm{diag}(\eta) r$ is free of the nuisance.
Write $W := \mathrm{span}\{\mathrm{diag}(\eta) v : \eta \in \mathcal G,\,
v \in \mathcal V\}$ for the orbit span. Then:
\emph{(a)} $W \subseteq \mathcal Z$ and $\mathrm{diag}(\eta) W = W$ for
every $\eta \in \mathcal G$;
\emph{(b)} if the alternating vector $z_C$ of a simple cycle $C$ lies in
$W$, then every $\eta \in \mathcal G$ is constant on $E(C)$;
\emph{(c)} let $\sim_W$ be the equivalence relation on $E$ generated by
declaring $e \sim_W e'$ whenever $e, e'$ lie on a common simple cycle $C$
with $z_C \in W$, and let $K$ be the number of $\sim_W$-classes meeting
$\bigcup_{w \in W} \mathrm{supp}(w)$. Then the image of $\mathcal G$
acting on $W$ has order at most $2^{K}$, and the capture of any orthonormal
contrast system contained in $W$ (in particular, any support system) is at most
$\|\Pi_W x\|^2 / V_n$.
Taking $\mathcal V$ to span $\mathcal Z$ recovers Lemma~\ref{lem:maxgroup}:
then $K = B$ and the bound $2^B$ is attained.

\begin{proof}
(a) Each generator lies in $\mathcal Z$ by admissibility, and $\mathcal Z$
is a subspace, so $W \subseteq \mathcal Z$. For $\eta \in \mathcal G$,
$\mathrm{diag}(\eta)$ maps the generating set $\{\mathrm{diag}(\sigma) v\}$
onto $\{\mathrm{diag}(\eta\sigma) v\}$, which is the same set because
$\eta \mathcal G = \mathcal G$ by closure; hence $\mathrm{diag}(\eta) W = W$.
(b) By (a), $\mathrm{diag}(\eta) z_C \in W \subseteq \mathcal Z$ and is
supported in $E(C)$; since $\mathcal Z \cap \R^{E(C)} = \mathrm{span}(z_C)$
as in the proof of Lemma~\ref{lem:maxgroup}, $\mathrm{diag}(\eta) z_C =
\pm z_C$, and non-vanishing of $z_C$ on $E(C)$ forces $\eta$ constant
there.
(c) By (b), $\eta$ is constant on every cycle generating $\sim_W$, hence
constant on each $\sim_W$-class. Two patterns agreeing on all classes that
meet the union of supports of $W$ act identically on $W$, so the action
factors through one sign per such class, giving at most $2^K$ distinct
elements. The capture bound is Bessel's inequality applied to the subspace
$W$. For $\mathcal V$ spanning $\mathcal Z$ we have $W = \mathcal Z$, every
simple cycle contributes its $z_C$, and by the ($\Rightarrow$) direction of
Lemma~\ref{lem:maxgroup} two edges are $\sim_W$-equivalent exactly when
they share a nontrivial block, so $K = B$.
\end{proof}

\subsection{Proof of Proposition~\ref{prop:dom} (Dominance)}

\emph{Restatement.} \emph{(a)} For any support system, the contrasts $v_j$ of
Definition~\ref{def:supp} satisfy the hypotheses of
Theorem~\ref{thm:exact}: each $v_j \in \mathcal Z$, supports are
disjoint, and each score $v_j'\varepsilon$ is symmetric under independent
symmetric errors with arbitrary heteroskedasticity. The associated
sign-flip test is exact at every sample size.
\emph{(b)} If $\mathcal C$ is an edge-disjoint cycle family and
$A_c = \mathrm{supp}(z_c)$, then for each $c$,
$(z_c'x)^2/\|z_c\|^2 \le \|\Pi_{\mathcal Z_{A_c}} x\|^2$, so
$\kap(\mathcal C) \le \kap(A_1, \ldots, A_C)$: replacing each packed
cycle vector by the local projection weakly increases capture, term by
term, at no cost in validity or in the number of contrasts.
\emph{(c)} Merging two supports, $A' = A_1 \cup A_2$, weakly increases
total capture:
$\|\Pi_{\mathcal Z_{A'}} x\|^2 \ge \|\Pi_{\mathcal Z_{A_1}} x\|^2 +
\|\Pi_{\mathcal Z_{A_2}} x\|^2$, at the cost of one contrast.

\begin{proof}
(a) A linear combination of independent symmetric random variables is
symmetric; disjoint supports give independence across $j$; membership in
$\mathcal Z$ is by construction. Theorem~\ref{thm:exact} applies
verbatim.
(b) $z_c/\|z_c\|$ is one unit vector in $\mathcal Z_{A_c}$; the
projection norm is the maximum of $(v'x)^2$ over unit $v \in
\mathcal Z_{A_c}$.
(c) $\mathcal Z_{A_1} \oplus \mathcal Z_{A_2} \subseteq \mathcal Z_{A'}$
(orthogonal, disjoint supports), and projection norms are monotone in the
subspace.
\end{proof}

\subsection{Proof of Theorem~\ref{thm:blockproj} (Block-projection test)}

\emph{Restatement.} Let $A_b = E_b$, $b = 1, \ldots, B$, be the nontrivial blocks of $G$, define
$\mathcal I_x:=\{b:\Pi_{\mathcal Z_b}x\ne0\}$ and $B_x:=|\mathcal I_x|$, and
for $b\in\mathcal I_x$ let
$v_b = \Pi_{\mathcal Z_b} x / \|\Pi_{\mathcal Z_b} x\|$. Then:
\emph{(a)} the sign-flip test with contrasts $(v_b)$ is exact at every
sample size under Assumption~\ref{ass:sym} with singleton blocks (independent,
symmetric, arbitrarily heteroskedastic errors);
\emph{(b)} its capture is $\kap = 1$: by \eqref{eq:blockdecomp},
$\sum_{b\in\mathcal I_x} (v_b'x)^2 = \sum_{b=1}^B \|\Pi_{\mathcal Z_b} x\|^2 = V_n$;
\emph{(c)} its effective score-flip group is $\{\pm 1\}^{B_x}$. It is induced
by the block-constant edgewise group $\{\pm 1\}^{B}$, which by
Lemma~\ref{lem:maxgroup} is the largest group of \emph{diagonal, edgewise}
sign flips preserving all of $\mathcal Z$; flips on omitted zero-capture blocks
act trivially on the statistic. Its capture is maximal in
the sense that $\kap \le 1$ for every orthonormal contrast system inside
$\mathcal Z$ (Bessel), a bound (b) attains. Two disclaimers belong with
(c). First, the full group $\{\pm 1\}^B$ is not maximal among exact
diagonal sign-flip tests as such---splitting a block into disjoint
supports can yield a larger group, with capture bounded by the smaller orbit
span of Proposition~\ref{prop:orbitmax}(c). The loss is strict for generic $x$
but can be zero for a treatment whose cycle-space projection already lies in
that span. Thus (c) separately asserts maximal capture and maximal preservation
of the full cycle space. Second, maximality of the \emph{group} is asserted only within
the diagonal class; see Remark~\ref{rem:scopeimposs}.

\begin{proof}
(a) is Proposition~\ref{prop:dom}(a) with $A_b = E_b$ for $b\in\mathcal I_x$. (b) is
\eqref{eq:blockdecomp}; omitted terms are zero. For (c), independent signs act
on the $B_x$ retained scores, while the other $B-B_x$ block signs lie in the
kernel of that action. The full edgewise statement is Lemma~\ref{lem:maxgroup}, plus
$\kap \le 1$ (Bessel, as after Definition~\ref{def:kappa}); the negative
half is Proposition~\ref{prop:orbitmax}(c). A witness inside the
bipartite class (not claimed minimal; e.g.\ $K_{2,4}$ is smaller) is the three-dimensional cube graph $Q_3$: it is a single
block, so its $\mathcal Z$-preserving flip group is $\{\pm 1\}$ with
$\kap = 1$, whereas the support system formed by its two opposite faces---
edge-disjoint four-cycles---is admissible with a group acting as
$\{\pm 1\}^2$, at capture $\|\Pi_W x\|^2 / V_n < 1$ for generic $x$, since
$\dim W = 2 < 5 = \dim \mathcal Z$.
\end{proof}

\section{The flip group is maximal}\label{app:maxgroup}

This section gives in full the material summarized in Section~\ref{sec:flipmax}.

Throughout, a \emph{block} of the multigraph $G$ is a maximal
$2$-connected subgraph (biconnected component), computable in linear time
by depth-first search; a pair of parallel edges (digon) is a block, and a
bridge is a trivial block containing no cycle. Let $\mathcal B_1, \ldots,
\mathcal B_B$ denote the nontrivial blocks, with edge sets $E_1, \ldots,
E_B$; these are pairwise edge-disjoint, and every cycle of $G$ has its
edges inside a single block, so the cycle space decomposes orthogonally:
\begin{equation}\label{eq:blockdecomp}
\mathcal Z \;=\; \bigoplus_{b=1}^{B} \mathcal Z_b,
\qquad \mathcal Z_b := \mathcal Z \cap \R^{E_b},
\qquad \Pi_{\mathcal Z} = \sum_b \Pi_{\mathcal Z_b},
\qquad V_n = \sum_{b=1}^B \|\Pi_{\mathcal Z_b} x\|^2 ,
\end{equation}
the last equality by Proposition~\ref{prop:cyclespace}(c) and
orthogonality of subspaces with disjoint supports.

\begin{lemma}[Maximality of the block-flip group among diagonal sign transformations]\label{lem:maxgroup}
For $\eta \in \{\pm 1\}^E$, the map $\mathrm{diag}(\eta)$ satisfies
$\mathrm{diag}(\eta)\,\mathcal Z \subseteq \mathcal Z$ if and only if
$\eta$ is constant on the edge set of every nontrivial block of $G$
(its values on bridges being unrestricted). Consequently the group
\[
\mathcal G_{\max} := \{\eta : \mathrm{diag}(\eta)\mathcal Z \subseteq \mathcal Z\}
\big/ \{\eta \equiv 1 \text{ on all cycles}\}
\;\cong\; \{\pm 1\}^{B},
\]
one flip per nontrivial block. It is therefore the largest group of
\emph{diagonal, edgewise} $\pm 1$ transformations that leaves every
contrast in $\mathcal Z$ nuisance-free;
Proposition~\ref{prop:orbitmax} gives the corresponding statement for flip
groups required only to respect the particular contrast system in use.
The scope of the claim should be read literally: it concerns the class
$\{\mathrm{diag}(\eta) : \eta \in \{\pm 1\}^E\}$ and nothing else. It says
nothing about randomization groups acting non-diagonally (rotations,
permutations, block-orthogonal maps), about conditional procedures, or
about exact tests not built from a group at all; see
Remark~\ref{rem:scopeimposs} and Remark~\ref{rem:openimposs}.
\end{lemma}

\begin{proof}
See Appendix~\ref{app:proofs}.
\end{proof}

\begin{remark}[What remains genuinely open on the impossibility side]\label{rem:openimposs}
Lemma~\ref{lem:maxgroup} closes the door on richer diagonal flip
groups, and only on those; Remark~\ref{rem:scopeimposs} lists what is left
outside. It does not prove that no exact test whatsoever extracts more
than block granularity: non-diagonal group actions, conditional procedures
(conditioning on cross-support statistics) and non-group constructions are
not covered. We conjecture
that any test with exact level $\alpha$ uniformly over independent
symmetric heteroskedastic errors has local power factoring through the
block-flip orbit statistics---a Lehmann--Stein--type completeness
statement---but we do not have a proof, and the convolution structure of
the observable $\Pi_{\mathcal Z}\varepsilon$ across a block makes the
completeness argument delicate.
\end{remark}

Lemma~\ref{lem:maxgroup} bounds the flip groups that preserve the whole
cycle space. A test, however, need only keep the contrasts it actually
uses free of the nuisance, and a group adapted to a smaller contrast
system can be larger. The following proposition states what survives at
that level of generality; it is the precise form of the informal claim
that sign-flip groups act block-constantly.

\begin{proposition}[Flip groups adapted to a contrast system]\label{prop:orbitmax}
Let $\mathcal V \subset \mathcal Z$ be a finite contrast system and let
$\mathcal G \subseteq \{\pm 1\}^E$ be a group of sign patterns
\emph{admissible} for $\mathcal V$, meaning $\mathrm{diag}(\eta) v \in
\mathcal Z$ for all $\eta \in \mathcal G$ and $v \in \mathcal V$, so that
every flipped statistic $v'\mathrm{diag}(\eta) r$ is free of the nuisance.
Write $W := \mathrm{span}\{\mathrm{diag}(\eta) v : \eta \in \mathcal G,\,
v \in \mathcal V\}$ for the orbit span. Then:
\emph{(a)} $W \subseteq \mathcal Z$ and $\mathrm{diag}(\eta) W = W$ for
every $\eta \in \mathcal G$;
\emph{(b)} if the alternating vector $z_C$ of a simple cycle $C$ lies in
$W$, then every $\eta \in \mathcal G$ is constant on $E(C)$;
\emph{(c)} let $\sim_W$ be the equivalence relation on $E$ generated by
declaring $e \sim_W e'$ whenever $e, e'$ lie on a common simple cycle $C$
with $z_C \in W$, and let $K$ be the number of $\sim_W$-classes meeting
$\bigcup_{w \in W} \mathrm{supp}(w)$. Then the image of $\mathcal G$
acting on $W$ has order at most $2^{K}$, and the capture of any orthonormal
contrast system contained in $W$ (in particular, any support system) is at most
$\|\Pi_W x\|^2 / V_n$.
Taking $\mathcal V$ to span $\mathcal Z$ recovers Lemma~\ref{lem:maxgroup}:
then $K = B$ and the bound $2^B$ is attained.
\end{proposition}

\begin{proof}
See Appendix~\ref{app:proofs}.
\end{proof}

Part (c) is the granularity trade-off in its sharpest form, and it cuts
both ways. A flip group can be made larger than $2^B$---by
splitting a block into several supports, so that $W$ no longer contains
cycles joining them---but only by shrinking $W$, and with it the capture
ceiling $\|\Pi_W x\|^2 / V_n$. Taking $W=\mathcal Z$ maximizes that ceiling
uniformly over treatments and forces the effective group down to $2^B$.
A proper $W$ can nevertheless capture a particular $x$ fully when
$\Pi_{\mathcal Z}x\in W$, so the realized loss is generic rather than
universal. Where to sit on this trade-off is precisely the splitting problem of
Section~\ref{sec:residual}.

The theta graph makes the lemma concrete. Take workers $a_1, a_2$ and
firms $b_1, b_2, b_3$ with edges forming three internally disjoint
$a_1$--$a_2$ paths, so three cycles pairwise sharing edges and
$\dim \mathcal Z = 2$. The whole theta is one block; by the lemma the
only computable flips are global, and indeed a direct computation shows
$\mathrm{diag}(\eta) z \in \mathcal Z$ for both basis cycles forces
$\eta$ constant on all six edges. Any edge-disjoint packing extracts one
cycle of the three and discards two: with $\pm 1$ cycle contrasts, up to
half the theta's capture can be lost. The next observation is that
nothing forces $\pm 1$ contrasts.

\section{Efficiency under heteroskedasticity}\label{app:het}

Theorem~\ref{thm:power} imposes a common error variance, whereas
Theorem~\ref{thm:exact} requires none. The gap matters for how $\kap$
should be read: it is the efficiency of the \emph{matched-weight}
contrast test against a \emph{homoskedastic} oracle, and neither the
optimal statistic nor the benchmark survives unchanged when the errors are
heteroskedastic. We record the general form, which is no harder.

Let $\Omega := \mathrm{diag}(\sigma_1^2, \ldots, \sigma_n^2)$ with
$0 < \sigma_{\min}^2 \le \sigma_i^2 \le \sigma_{\max}^2 < \infty$, and set
\[
\omega_c^2 \;:=\; \Var(v_c'\varepsilon) \;=\; v_c'\Omega v_c
\;=\; \sum_i v_{ci}^2 \sigma_i^2 ,
\qquad
\mathcal{E}_n \;:=\; \sum_{c \le C_n} b_c^2/\omega_c^2 .
\]
The natural benchmark is now the generalized least squares oracle: the
infeasible test based on the efficient score for $\beta$ when $\Omega$ is
known and $\gamma$ is not, whose efficacy is
\[
\mathcal{V}_\Omega \;:=\; \big\| M_{\Omega^{-1/2}D}\, \Omega^{-1/2} x \big\|^2
\;=\; x'\Omega^{-1}x - x'\Omega^{-1}D(D'\Omega^{-1}D)^{-}D'\Omega^{-1}x .
\]
When $\Omega = \sigma^2 I$ this is $V_n/\sigma^2$ and
$\mathcal{E}_n = \Sigma_n/\sigma^2$.

\begin{proposition}[Optimal weighting and heteroskedastic efficiency]\label{prop:het-supp}
Let blocks be singletons, the errors independent, symmetric, with
covariance $\Omega$ as above and $\sup_i \E[\varepsilon_i^4] \le \bar\kappa$.
For a weight vector $a = (a_c)_{c \le C_n}$ not identically zero (individual
$a_c$ may be zero) with $\sum_c a_cb_c\ne0$ --- error-independent, but possibly
depending on $\Omega$ as well as the design, hence infeasible unless $\Omega$
is known, as for the optimal weights of part~(a) --- let $T_a(u) = |\sum_c a_c u_c|$,
which is exact by Theorem~\ref{thm:exact} for every choice of $a$, and let
\[
e_n(a) \;:=\; \Big(\sum_c a_c b_c\Big)^2 \Big/ \sum_c a_c^2 \omega_c^2
\]
be its efficacy. Then:
\begin{enumerate}[label=(\alph*)]
\item $e_n(a) \le \mathcal{E}_n$ for every $a$, with equality if and only if
$a_c \propto b_c/\omega_c^2$.
\item Suppose $e_n(a)\to\infty$ and both
\begin{equation*}\tag{P1a}\label{eq:P1a}
\max_{c \le C_n} a_c^2\omega_c^2 \Big/ \sum_c a_c^2 \omega_c^2
\;\longrightarrow\; 0
\end{equation*}
and
\begin{equation*}\tag{P1b}\label{eq:P1b}
\sum_c a_c^2 b_c^2 \Big/ \Big(\sum_c a_c b_c\Big)^{2}
\;\longrightarrow\; 0 .
\end{equation*}
Then along $\beta_n = \beta_0 + h\, e_n(a)^{-1/2}$ the power of the
level-$\alpha$ test based on $T_a$ converges to
$\Phi(h - z_{1-\alpha/2}) + \Phi(-h - z_{1-\alpha/2})$.
Condition \eqref{eq:P1b} is not implied by \eqref{eq:P1a} and cannot be
dropped: see Remark~\ref{rem:p1b}. For the optimal weights of part~(a),
\eqref{eq:P1a} implies \eqref{eq:P1b}; and when $\Omega = \sigma^2 I$ and
$a_c = b_c$, both reduce to \eqref{eq:P1}, so Theorem~\ref{thm:power}(i)
is the special case it appears to be.
\item $\mathcal{E}_n \le \mathcal{V}_\Omega$, so
\[
\kap^{\mathrm{het}} \;:=\; \mathcal{E}_n / \mathcal{V}_\Omega \;\in\; [0,1] ,
\]
with equality if and only if $M_{\Omega^{-1/2}D}\Omega^{-1/2}x$ lies in
$\mathrm{span}\{\Omega^{1/2}v_c/\omega_c\}_c$. If moreover $\mathcal E_n\to\infty$,
$\max_i (M_{\Omega^{-1/2}D}\Omega^{-1/2}x)_i^2 / \mathcal{V}_\Omega \to 0$ and
$\kap^{\mathrm{het}}$ converges, then $\lim \kap^{\mathrm{het}}$ is the Pitman
asymptotic relative efficiency of the optimally weighted contrast test
with respect to the GLS oracle. Under homoskedasticity
$\kap^{\mathrm{het}} = \kap$.
\end{enumerate}
\end{proposition}

\begin{proof}
See Appendix~\ref{app:proofs}.
\end{proof}

\begin{remark}[Why \eqref{eq:P1b} is needed]\label{rem:p1b}
Conditions \eqref{eq:P1a} and \eqref{eq:P1b} constrain different objects:
the first says no contrast dominates the \emph{null} dispersion of the
statistic, the second that none dominates its \emph{signal}. When weights
are matched to loadings the two coincide, which is why a single condition
suffices in Theorem~\ref{thm:power}; for a general $a$ they can come
apart badly. Take $\omega_c \equiv 1$, equal weights $a_c \equiv 1$, and a
loading vector concentrated on one contrast, $b_1 = C_n$ and $b_c = 0$ for
$c \ge 2$---a packing in which a single cycle carries all the treatment
variation and the rest are noise. Then $\max_c a_c^2\omega_c^2/A_n = 1/C_n
\to 0$, so \eqref{eq:P1a} holds, while
$\sum_c a_c^2b_c^2/(\sum_c a_cb_c)^2 = 1$ for every $C_n$, so
\eqref{eq:P1b} fails. Here $e_n(a) = C_n\to\infty$, so the local alternative is
$\delta_n = h C_n^{-1/2}$, and the single score has signal
$b_1\delta_n=hC_n^{1/2}$. Conditional on the data, global sign reversal of
that dominant coordinate makes the limiting absolute randomization law the
same $|N(h,1)|$ law as the observed statistic. The test therefore has no
asymptotic power beyond size: power converges to $\alpha$, not to the limit
in~(b)---at $h = 3$ and
$\alpha = 0.05$, $0.85$ is claimed and $0.05$ obtains. This is the same
one-fat-coordinate degeneracy as Remark~\ref{rem:p1}, arriving through the
signal rather than through the noise, and it is verified numerically in
\texttt{verification/verify\_power\_conditions.py}.
\end{remark}

\begin{remark}[Feasibility and the packing objective]\label{rem:hetfeas}
The optimal weights $b_c/\omega_c^2$ are infeasible, and with one contrast
score per cycle the $\omega_c^2$ are not consistently estimable
individually. This is a pure efficiency question, not a validity one:
Theorem~\ref{thm:exact} holds for every fixed
$T : \R^{C} \to \R$, including statistics that weight the coordinates
using $|u_1|, \ldots, |u_C|$, since such a $T$ is still a fixed function of
its argument and the sign-flip group still acts on $u$. Adaptive
weighting therefore cannot break exactness and can only be judged on
power. A common sign-invariant normalizer, as in Remark~\ref{rem:stat},
leaves the randomization $p$-value unchanged; changes to relative coordinate
weights can change power.

Two consequences for what follows. First, $\kap$ and
$\kap^{\mathrm{het}}$ are not ordered in general, so the empirical values
reported in Section~\ref{sec:kappa} should be read as the homoskedastic
benchmark rather than as an efficiency claim valid under the
heteroskedasticity that Theorem~\ref{thm:exact} tolerates. Second, the
packing problem of Section~\ref{sec:packing} maximizes $\kap$; by
Proposition~\ref{prop:het-supp}(c) that is the correct objective under
homoskedasticity, and under heteroskedasticity the correct objective is
$\mathcal{E}_n$, which reweights each candidate cycle by its own
$\omega_c^{-2}$ and so favours cycles through low-variance observations.
Our algorithm targets $\kap$; we regard the weighted packing problem as
open.
\end{remark}

\noindent The remaining appendices collect material that supports the paper without being needed to follow it: a detailed comparison with existing methods, the block structure of the two real designs, Monte Carlo evidence on dependence, the non-canonicity of the contrast system, and a diffuse dense-panel application. Every proof, together with the flip-group maximality results and the heteroskedastic efficiency theory, is in Appendices~\ref{app:proofs}--\ref{app:het} above.

\section{Detailed comparison with existing methods}

This section expands the summary given in Section~\ref{sec:related}.

The direct antecedent for our default statistic is \citet{Pouliot2026}.
His exact $t$-test projects the tested regressor onto the orthogonal complement
of the span of all group-transformed nuisance regressors. With the group chosen
as independent sign flips on disjoint supports, this projector is block
diagonal. The projection-contrast matched score is therefore exactly his
unstudentized statistic on the selected observations, as
Proposition~\ref{prop:exactt} shows. The rank is the sum of the local residual
dimensions, not a worst-case rank bound based on the number of global FE columns.
Our cycle-space analysis, support construction, and capture calculations make
that general method operational in two-way FE designs.

Pouliot also develops exact inference under invariance with asymptotic
robustness when invariance fails. His block-permutation asymptotic theorem has
different hypotheses from a support-sign test in a growing FE design.
Proposition~\ref{prop:asymmetry} therefore gives the sign-flip argument directly,
allowing asymmetric independent centered contrast scores with heterogeneous
variances under a fourth-moment negligibility condition. Its self-normalizer is
constant on each sign orbit, so the existing matched-weight $p$-values already
have this additional large-sample interpretation when those conditions hold.
This does not remove the symmetry requirement for finite-sample exactness or
supply an asymmetry guarantee for persistently dominant score contributions.

The few-cluster literature provides related sign-change procedures. \citet{CRS2017} obtain asymptotic validity of sign-change randomization when a fixed number of group-level statistics are asymptotically normal; \citet{CSS2021} show the wild cluster bootstrap is such a test and require homogeneity restrictions on the covariate distribution across clusters; \citet{Toulis2022} requires an analogous homogeneity condition for cluster sign flips with a fixed number of clusters. All three are asymptotic schemes in which nuisance parameters are estimated and the invariance holds only in the limit. That mode of argument goes back to \citet{Romano1990}, who characterizes when a randomization test retains asymptotic level without an exact group invariance. In our construction nothing is estimated---$q_c'D = 0$ holds exactly, by design---so Theorem~\ref{thm:exact} is finite-sample and homogeneity-free; the trade is the explicit symmetry in Assumption~\ref{ass:sym}, which \citet{DutzZhang2026} indicate cannot be substantially weakened while preserving exactness. The sign-flipped score tests of \citet{HemerikGoemanFinos2020} and \citet{DeSantisEtAl2025} are the closest statistical relatives; there, flipping estimated score contributions yields asymptotically exact tests, with corrections for nuisance estimation, whereas annihilation removes the nuisance before flipping. Finally, \citet{IbragimovMuller2010} offer a different few-group route via $t$-statistics on group estimates, which requires each group to identify $\beta$ on its own; cycle contrasts can combine observations from several groups. On the network side, \citet{JochmansWeidner2019} characterize how graph connectivity governs the accuracy of estimated fixed effects, and \citet{KSS2020} and \citet{Jochmans2022} develop leave-out variance estimation under diffuse-score asymptotics; our results are complementary, operating exactly when those asymptotics fail. \citet{Crippa2025} uses matching-network cycles to test the additive TWFE specification itself; we take the specification as given and use cycles for exact inference on $\beta$.

Other exact tests address a scalar coefficient in fixed-design linear models with nuisance regressors. \citet{WenWangWang2025} project residuals onto the orthogonal complement of the union of the original and permuted design spaces, obtaining finite-sample validity under exchangeable noise when $p<n/2$. Saturation can violate this condition, but does not do so in all our examples: the Grunfeld and F-score designs have full regression ranks $32$ and $62$ with $n=220$ and $217$, respectively, so even their full ranks are below $n/2$. The dimensional requirement is weaker than that of \citet{LeiBickel2021}. \citet{DHaultfoeuilleTuvaandorj2024} and \citet{Tuvaandorj2024} develop permutation tests that are exact under independence between the tested regressor and the remaining regressors (respectively, between instruments and structural errors) and asymptotically valid, allowing heteroskedasticity, when that independence is relaxed; \citet{DiCiccioRomano2017} is the antecedent. Our fixed-design formulation does not assume independence between the treatment and nuisance regressors; exactness instead follows from the stated error invariance after nuisance annihilation.

 \citet{LiZhouZhang2026} place the CPT/PALMRT lineage in an explicit group framework: any finite group of permutation matrices yields, under fully exchangeable noise, Type I control at level $2\alpha$ for grouped permutation-augmented tests, and they exhibit worst-case designs showing the factor two is unimprovable; power is analyzed through a spectral separation between the target regressor and its permuted projections, optimized by a design-adaptive choice of group under sub-Gaussian design assumptions, and their extension beyond exchangeability proceeds by weighted-conformal arguments, with control degrading in total-variation distance from group symmetry. Their invariance is permutational and requires exchangeability, which heteroskedasticity breaks; ours is a sign-flip (reflection) invariance requiring symmetry but tolerating arbitrary heteroskedasticity, and delivering rejection probability at most $\alpha$ with no degradation term under the maintained symmetry.

Their group optimization searches subgroups of the symmetric group for a generic design, while our packing problem (Section~\ref{sec:packing}) is posed on a specific object---the cycle space of the bipartite design multigraph---where digons, disjoint firm-pair four-cycles (Remark~\ref{prop:auto}), and contraction (Proposition~\ref{prop:contract}) provide a structured packing heuristic. In both settings, the choice of group affects power; Section~\ref{sec:overlap} characterizes this trade-off for the sign-flip groups considered here. The general principle that a well-chosen subgroup can dominate the full group \citep{KoningHemerik2024, RamdasEtAl2023} applies to both.

On the network side, beyond \citet{JochmansWeidner2019}, \citet{KSS2020}, and \citet{Jochmans2022}, two recent contributions bear on the same designs from other directions: \citet{Sakamoto2025} studies inference on the fixed effects themselves under node- and edge-level dependence, where least squares can be inconsistent and a Conley--Taber-style procedure restores asymptotic validity; and \citet{ChengHoSchorfheide2025} model limited mobility spectrally, through Laplacian eigenvalues tending to zero, and propose an empirical Bayes estimator. Both develop asymptotic methods for fixed effects and their moments; the target here is $\beta$.

\section{Block structure of real designs}

This section gives in full the material summarized in Section~\ref{sec:residual}.

\subsection*{The formal statement of the splitting problem}

\begin{center}
\emph{Given a biconnected block $\mathcal B$ with local cycle space
$\mathcal Z_{\mathcal B}$ and target count $k$, partition $E(\mathcal B)$
into supports $A_1, \ldots, A_k$ maximizing
$\sum_j \|\Pi_{\mathcal Z_{A_j}} x\|^2$, subject to a balance constraint
$\max_j \|\Pi_{\mathcal Z_{A_j}} x\|^2 \le \tau \sum_j \|\Pi_{\mathcal Z_{A_j}} x\|^2$.}
\end{center}

Cycle packing (Section~\ref{sec:packing}) is the special case in which
each $A_j$ is required to be a single cycle and the projection is
replaced by the $\pm 1$ vector; Proposition~\ref{prop:dom}(b--c) shows
both restrictions only lose capture. The interpolation endpoints are
$k = 1$ (full capture, no granularity) and $k = \dim \mathcal Z_{\mathcal B}$
(cycle-basis granularity, capture bounded by the best disjoint packing).
Ear decompositions are the natural splitting tool: a biconnected graph is
an ear sequence, each ear adding one dimension to
$\mathcal Z_{\mathcal B}$, and peeling ears with small marginal capture
produces nested support systems along which
$\sum_j \|\Pi_{\mathcal Z_{A_j}} x\|^2$ can be tracked exactly; we do not
know sharp approximation guarantees and leave them open, together with
the heteroskedastic version in which each support's contrast maximizes
$(v'x)^2 / v'\Omega v$ (Remark~\ref{rem:hetfeas}).

\subsection*{Block structure of real designs}

The block decomposition is a linear-time computation, and on the two real
designs of Section~\ref{sec:numerics} it returns one verdict. On the
\citet{KSS2020} network the observation multigraph has $27{,}694$
nontrivial blocks---$27{,}564$ stayer digons and $130$ larger blocks, the
largest containing $12{,}356$ edges---and the identity
$\sum_b \|\Pi_{\mathcal Z_b} x\|^2 = V_n$ verifies to six digits for both
treatments, so the $\kap = 1$ of Theorem~\ref{thm:blockproj} is attained.
For the match-level treatment all $27{,}564$ stayer digons have zero capture,
so they are omitted from the block test and do not contribute orbit granularity;
all $130$ larger blocks are active, giving $B_x=130$. Moreover,
the giant block carries $93.7\%$ of
$V_n$ (effective number of contrasts $\approx 1$), and for the
time-varying covariate its share is $14.3\%$; on the real F-score panel of
Section~\ref{sec:kappa} the graph splits into exactly three nontrivial
blocks, one per country, the largest---Poland's, with $122$ of the $217$
firm-years---carrying $69.6\%$ of $V_n$ ($C_{\mathrm{eff}} = 1.9$).
The balance condition $\max_b\|\Pi_{\mathcal Z_b}x\|^2/V_n\to0$ therefore fails in
every case: full capture is achievable but concentrated, and the
splitting problem above---not overlap---is the binding constraint on
realistic designs. Consistently, the within-support upgrades of
Proposition~\ref{prop:dom} are nearly exhausted on these networks: on the
KSS match-level packing, firm-pair projection contrasts raise $\kap$ from
$0.509$ to $0.523$ while concentrating loadings (maximum
$b_c^2$-share $0.007 \to 0.168$), and on the real F-score panel the same
move lowers capture ($0.795 \to 0.723$, with the support count
falling from $45$ to $18$) because on a dense graph large supports compete
for edges---dominance holds per support and per merge
(Proposition~\ref{prop:dom}), not across changes of the support family,
so support granularity is itself a choice variable of the splitting
problem. Replication code for all of these computations accompanies the
paper.

\section{Monte Carlo evidence on dependence}

This section studies finite-sample size under alternative error-dependence structures on the F-score design described in Appendix~H.

\begin{table}[ht]
\centering
\caption{Empirical size at nominal $5\%$ under dependence, real F-score design, true null $\beta = 0$.}
\label{tab:depmc}
\small
\setlength{\tabcolsep}{4pt}
\begin{tabular}{@{}lcccc@{}}
\toprule
Error process & HC2 $t$ & clustered $t$ & obs.\ cycles ($C{=}45$) & firm groups ($C{=}8$) \\
\midrule
\multicolumn{5}{l}{\emph{Covered with singleton blocks after FE annihilation}}\\
\quad independent Laplace & 0.056 & 0.047 & 0.053 & 0.051 \\
\quad independent, heteroskedastic & 0.045 & 0.039 & 0.054 & 0.048 \\
\quad additive firm $+$ period shocks & 0.046 & 0.040 & 0.049 & 0.046 \\
\addlinespace
\multicolumn{5}{l}{\emph{Covered only with blocks $=$ firms}}\\
\quad within-firm AR(1), $\rho = 0.6$ & 0.040 & 0.041 & 0.038 & 0.048 \\
\quad within-firm AR(1), $\rho = 0.9$ & 0.046 & 0.041 & 0.041 & 0.047 \\
\addlinespace
\multicolumn{5}{l}{\emph{Covered with firm blocks conditional on factors}}\\
\quad interactive effects $\gamma_i f_t+e_{it}$ & 0.049 & 0.041 & 0.049 & 0.047 \\
\bottomrule
\end{tabular}
\par\smallskip
\begin{minipage}{0.95\textwidth}\footnotesize
Notes: $R = 3{,}000$ replications, $B = 999$ flips, symmetric innovations throughout; Monte Carlo standard error at $5\%$ is $0.004$, so a $\pm 3$ s.e.\ band is $[0.038, 0.062]$. ``obs.\ cycles'' is the $C = 45$ four-cycle system of Appendix~H; ``firm groups'' is the $C = 8$ firm-measurable system of the clustering-cost figures in that section. Row~3 illustrates Remark~\ref{prop:reann}: additive shocks at either fixed-effect level are annihilated by $q'D = 0$, so they are invisible to every contrast system regardless of magnitude. Rows~4--6 are not covered for the observation-level system. Row~6 is covered for firm-compatible supports conditional on the factor vector, as explained in Section~\ref{sec:dep}.
\end{minipage}
\end{table}

The reported frequencies are consistent with the finite-sample size control of Theorem~\ref{thm:exact} wherever its assumptions apply. For the interactive process, the independent symmetric loadings and independent centrally symmetric idiosyncratic firm vectors give independent centrally symmetric firm errors conditional on the factor vector, so firm-compatible sign flips are valid conditionally and hence unconditionally. The observation-level results under within-firm AR(1) and interactive effects remain evidence of robustness in this design, rather than a general validity result. Remark~\ref{prop:reann} explains the annihilation of additive components in the fixed-effect span; it does not remove a general AR(1) covariance.

The block condition matters in other designs. Consider a complete $F\times T$ panel with separable errors $\varepsilon_{ft}=g_fh_t$, $g_f$ symmetric and independent across firms and independent of $h$, so firm blocks are independent and centrally symmetric conditional on $h$. A four-cycle on firms $\{f,f'\}$ and periods $\{t,t'\}$ then has score $\tfrac12(g_f-g_{f'})(h_t-h_{t'})$: every four-cycle on the same firm pair carries the same factor $g_f-g_{f'}$, so those contrasts are deterministically proportional and independent sign flips do not preserve their joint distribution. At $F=4$, $T=6$ the four-cycle test has size $0.22$ at nominal $10\%$ while the firm-grouped test remains valid; at $F=12$ the same DGP gives $0.15$, and by the scale of Table~\ref{tab:depmc}'s design the distortion has diluted below detection. Robustness there is thus a property of that design, not of the method, whereas Assumption~\ref{ass:sym}(iii) guarantees validity irrespective of it (\texttt{verification/verify\_new\_results.py}).

For the design in Table~\ref{tab:depmc}, firm grouping accommodates within-firm dependence while retaining $\kap=0.804$, compared with $0.795$ for observation-level supports. It reduces the number of signs from $45$ to $8$ (Appendix~H).

\section{The contrast system is not canonical}

Every object in Appendix~H below the design diagnostics is a property of the
particular contrast system used, not of the design. Weighted edge-disjoint cycle
packing is NP-hard, so the algorithm of Section~\ref{sec:packing} is a heuristic; and when the
treatment is discrete, as the F-score is, many four-cycles carry identical
loadings, so a greedy rule also needs a tie convention. Canonical design-label
ordering makes that convention reproducible, but it does not make the resulting
heuristic a mathematically unique optimum.

The current v0.6.0 implementation ranks candidates on the raw treatment, as the
identity $v_c'\tx=v_c'x$ permits, and resolves every tie by canonical design
labels. It returns $C=43$ supports covering $174$ of the $217$ observations,
with $\kappa_C=0.827$. The support/sign signatures are identical in the Julia
and R implementations and unchanged over input-row permutations. An earlier
archived construction returned $C=45$ and $\kappa_C=0.795$. Both systems are
admissible under Theorem~\ref{thm:exact}, so both tests are exact, and the verdict
is unchanged: the exact set fails to exclude zero, at $p=0.72$ under the current
packing and $p=0.40$ under the archived packing.

The current packing gives $\widetilde\beta=-0.006$ against the archived $+0.015$, with exact set $[-0.040,\,0.026]$ against $[-0.022,\,0.051]$. Under the maintained common-slope model, both contrast estimators target the same $\beta$:
\[
 \widetilde\beta
 =\frac{\sum_c b_c v_c'Y}{\sum_c b_c^2}
 =\beta+\frac{\sum_c b_c v_c'\varepsilon}{\sum_c b_c^2}.
\]
Changing a design-only packing changes the weights on the noise, not the target parameter. If the errors have zero means and the relevant expectations exist, this estimator is conditionally unbiased; moment-free symmetry alone does not guarantee existence of its expectation. The sign change between these estimates does not establish a change of estimand, and failure to reject zero does not establish a zero effect. Report the confidence set alongside the point estimate to show the uncertainty for the selected admissible system.

The same applies to the symmetry diagnostic, and there the direction is less
comfortable. Assumption~\ref{ass:sym} restricts the contrast scores, and the scores are
chosen rather than given: the raw-return contrast-score skewness is $-1.08$
under the current packing but $+0.30$ under the archived one. The former is
roughly $2.9$ times the normal-i.i.d.\ reference scale $\sqrt{6/43}=0.37$;
this is a descriptive heuristic, not a calibrated symmetry test under
heteroskedastic scores. The main paper accordingly
reports the unfavorable current value and does not claim that differencing
establishes symmetry. The check must be recomputed for whichever system is
actually used. Neither reading disturbs
the exactness of the test, which requires only symmetry of the scores that are
used; but a practitioner who selects a packing, finds the symmetry diagnostic
unfavourable, and then selects another has spent a researcher degree of freedom,
and should report the fact.

\section{A diffuse dense panel}

We finally run a second specification end to end on the real F-score panel.

The specification is the two-way saturated regression
\begin{equation}\label{eq:emp}
r_{it} \;=\; \beta\, F_{it} \;+\; \phi_i \;+\; \psi_{c(i)t} \;+\; \varepsilon_{it},
\end{equation}
where $F_{it}$ is the Piotroski-type F-score \citep{Piotroski2000}, with firm and
country--year effects. The design has $217$ observations, $64$ dummy columns spanning $61$ fixed-effect dimensions,
and $155$ residual degrees of freedom. Log returns are the headline; raw
returns provide a robustness column.

\emph{Step 1: design diagnostics.} Before viewing the outcome, $V_n=461.69$,
$\lambda_n=0.036$, and $N_{\mathrm{eff}}=83.1$ identify a diffuse design in
which exactness pays a visible price rather than rescuing invalid Gaussian inference.

\emph{Step 2: contrasts and capture.} The deterministic v0.6.0 design-only
packing returns $43$ cycle contrasts ($42$ four-cycles and one six-cycle)
covering $174$ observations, $\kap=0.827$,
an SE price of $1.10\times$, and maximum loading share $0.065$, consistent with the diffuseness required by (P1).

\emph{Step 3: estimates and intervals.} Table~\ref{tab:emp} collects them.

\begin{table}[ht]
\centering
\caption{Complete workflow for specification~\eqref{eq:emp} on the real F-score panel ($n = 217$; $19$ firms; $3$ countries; $2010$--$2024$).}
\label{tab:emp}
\small
\setlength{\tabcolsep}{2pt}
\begin{tabular}{@{}>{\raggedright\arraybackslash}p{3.4cm}cccc@{}}
\toprule
& \multicolumn{2}{c}{$\log(1+r)$} & \multicolumn{2}{c}{raw $r$} \\
\cmidrule(lr){2-3}\cmidrule(lr){4-5}
& estimate & 95\% interval & estimate & 95\% interval \\
\midrule
OLS $\hbeta$ & $0.0030$ & & $0.0122$ & \\
contrast $\widetilde\beta$ & $-0.0057$ & & $-0.0036$ & \\
\addlinespace
homoskedastic $t$ & & $[-0.0347,\, 0.0407]$ & & $[-0.0831,\, 0.1074]$ \\
HC2 $t$ & & $[-0.0309,\, 0.0369]$ & & $[-0.0563,\, 0.0807]$ \\
clustered by firm ($G=19$) & & $[-0.0353,\, 0.0412]$ & & $[-0.0547,\, 0.0790]$ \\
\textbf{exact sign-flip} & & $\mathbf{[-0.0393,\, 0.0264]}$ & & $\mathbf{[-0.0573,\, 0.0473]}$ \\
\midrule
exact $p$-value at $\beta_0 = 0$ & \multicolumn{2}{c}{$0.719$} & \multicolumn{2}{c}{$0.893$} \\
width relative to HC2 & \multicolumn{2}{c}{$0.97\times$} & \multicolumn{2}{c}{$0.76\times$} \\
$\lambda_n$ / $N_{\mathrm{eff}}$ (treatment) & \multicolumn{4}{c}{$0.036$ / $83.1$} \\
$\widehat\lambda_n^{\mathrm{score}}$ / $\widehat N_{\mathrm{eff}}^{\mathrm{score}}$ & \multicolumn{2}{c}{$0.133$ / $26.6$} & \multicolumn{2}{c}{$0.186$ / $9.9$} \\
within-residual skew / contrast-score skew & \multicolumn{2}{c}{$+0.49$ / $+0.68$} & \multicolumn{2}{c}{$+3.71$ / $-1.08$} \\
\bottomrule
\end{tabular}
\par\smallskip
\begin{minipage}{0.97\textwidth}\footnotesize
Notes: $C = 43$ cycle contrasts ($42$ four-cycles and one six-cycle), $\kap = 0.827$, $\kap^{-1/2} = 1.10$, in both columns (the contrast system depends on $(x,D)$ only). Exact intervals invert the test of Theorem~\ref{thm:exact} on a grid of $24{,}001$ values of $\beta_0$ with $B = 99{,}999$ flips and common random numbers; the accepted set is an interval in both columns. The feasible score diagnostics are defined in Remark~\ref{rem:diag}. The normal-i.i.d.\ reference scale for skewness at $C=43$ is $\sqrt{6/C}=0.37$; this is a heuristic, not a general standard error under heteroskedasticity.
\end{minipage}
\end{table}

Neither procedure rejects zero. The OLS estimate associates one F-score point
with $0.30$ percentage points of annual log return, and neither interval excludes zero
($p=0.87$ conventionally, $0.719$ exactly).

The contrast estimate $\widetilde\beta=-0.0057$ differs from OLS $0.0030$
because the selected supports weight the noise differently; both target $\beta$
under the common-slope model. The inverted exact interval need not be centered
on either estimate. Its width is $0.97\times$ the HC2 width.

Raw returns expose what treatment leverage misses: although $\lambda_n=0.036$,
$\widehat N_{\mathrm{eff}}^{\mathrm{score}}=9.9$. The adaptive exact interval is then
$24\%$ narrower than HC2. The symmetry diagnostic is unfavorable, however:
raw residual skewness is $+3.71$ and the selected contrast scores have skewness
$-1.08$ (normal-i.i.d.\ reference scale $0.37$). Thus this diffuse case illustrates the mechanism and the
diagnostic; it does not empirically validate Assumption~\ref{ass:sym}.
\subsection*{The cost of clustering}

Section~\ref{sec:dep} argued that the dependence structures a contrast system can
tolerate are decided by a combinatorial condition on its supports. On the same
panel that condition is a table (the figures above), and it exhibits the
capture--granularity trade-off in its most concrete form.

For comparability across dependence partitions, this subsection uses the
archived nested support family rather than the stronger unconstrained v0.6.0
packing used above. Its four levels give, in order of coarseness: observation-
or match-level blocks, $C = 45$ supports with $\kap = 0.795$ and maximum share
$0.068$; period (country--year) blocks, $C = 21$ and $\kap = 0.776$; firm
blocks, $C = 8$ and $\kap = 0.804$; and country blocks, $C = 3$ and
$\kap = 1.000$ with maximum share $0.696$. Each row is the finest support
system whose supports are unions of the stated blocks, built by greedily
merging blocks to maximize local capture (Definition~\ref{def:supp}). For the
default two-sided statistic the smallest full-enumeration $p$-value is
$2^{1-C}$, because $s$ and $-s$ produce the same absolute statistic, so a
level-$5\%$ full-enumeration test requires $C \ge 6$.

Coarsening a \emph{nested} support family allows larger dependence blocks and,
by Proposition~\ref{prop:dom}(c), weakly increases capture. The period and firm
partitions here are not nested, so their captures need not be ordered; along the
nested observation--firm--country comparison $\kap$ rises from $0.795$ to
$0.804$ and then to $1$. The corresponding contrast counts are $45$, then
$21$, then $8$, then $3$. At the bottom row the exact test attains full capture
and \emph{cannot reject at the $5\%$ level} under the stated nonrandomized full-enumeration rule, because a randomization group of
order $2^3$ cannot produce a two-sided full-enumeration $p$-value below $0.25$
(the generic, non-even orbit bound is $0.125$). The stated country-block sign-flip procedure therefore cannot reject at conventional levels on this three-country panel. This discrete-orbit obstruction differs from the few-group $t$ approach of \citet{IbragimovMuller2010}, which can be used with two or more independent asymptotically Gaussian group estimates and has no analogous $0.25$ minimum $p$-value with three groups. Both approaches can pay a power cost for few groups, but their restrictions are different. For the present sign-flip test, granularity must be considered alongside capture.

Firm-level clustering, by contrast, is comfortably feasible: $C = 8$ supports
give a group of order $256$ and $\kap = 0.804$, so a researcher who believes a
firm's return deviations are serially correlated pays essentially nothing in
capture and retains a valid $5\%$ test.

Appendix~F reports the dependence simulations and a counterexample with observation-level size $0.22$ at nominal $10\%$, while firm grouping remains valid. The firm-grouped guarantees include the conditional factor argument in Section~\ref{sec:dep}; observation-level robustness outside the stated assumptions remains design-specific.

\end{document}